\documentclass[journal,compsoc,10pt]{IEEEtran}

\usepackage[utf8]{inputenc}
\usepackage[T1]{fontenc}

\usepackage{mathptmx} 
\usepackage[scaled=.90]{helvet} 
\DeclareMathAlphabet{\mathsf}{T1}{phv}{m}{n} 
\usepackage{courier} 
\usepackage{microtype} 
\microtypesetup{expansion=true,protrusion=true,stretch=30,shrink=20,factor=1100}

\usepackage[hyphens]{url}

\PassOptionsToPackage{hyphens}{url}

\usepackage{cite}
\usepackage{amsmath,amssymb,amsfonts}
\usepackage{amsthm}
\usepackage{mathtools} 
\usepackage{algorithmic}
\usepackage{graphicx}
\usepackage{textcomp}
\usepackage{xcolor}
\usepackage{listings}
\usepackage{booktabs}
\usepackage{tabularx}        
\usepackage{makecell}        
\usepackage{seqsplit}        
\usepackage{ragged2e}        
\usepackage{enumitem}        
\usepackage{url}
\usepackage{hyperref}

\hypersetup{
 colorlinks=true,
 linkcolor=black,
 citecolor=black,
 urlcolor=blue,
 pdftitle={Token Budgets: An Empirical Catalog of 63 LLM-Agent Budget-Overrun Incidents, with an Affine-Typed Rust Mitigation as a Case Study},
 pdfauthor={Sajjad Khan},
 pdfsubject={Compile-time cost control for LLM agent systems via affine resource types},
 pdfkeywords={Empirical software engineering, LLM agents, Cost control, Affine types, Rust, Failure catalog}
}

\newtheorem{lemma}{Proposition}
\newtheorem{conjecture}{Conjecture}

\lstdefinelanguage{Rust}{
 keywords={fn,let,mut,pub,struct,enum,impl,trait,async,await,move,
 Self,self,true,false,if,else,match,return,for,in,while,loop,
 use,mod,as,where,type,const,static,unsafe,extern,crate,
 ref,break,continue},
 keywordstyle=\color{blue}\bfseries,
 commentstyle=\color{gray}\itshape,
 stringstyle=\color{teal},
 morecomment=[l]{//},
 morecomment=[s]{/*}{*/},
 morestring=[b]",
 basicstyle=\ttfamily\footnotesize,
 showstringspaces=false,
 breaklines=true,
 frame=single,
}
\begin{document}

\title{\sffamily\selectfont Token Budgets: An Empirical Catalog
of 63 LLM-Agent Budget-Overrun Incidents, with an Affine-Typed
Rust Mitigation as a Case Study}

\author{Sajjad~Khan%
\IEEEcompsocitemizethanks{%
\IEEEcompsocthanksitem S.~Khan is an independent researcher; MSc in Data
Science, University of the West of England, Bristol, UK.
E-mail: \texttt{sajjadanwar200@gmail.com}.
ORCID: \href{https://orcid.org/0009-0007-8627-7682}{0009-0007-8627-7682}.
Source code and full artifact bundle:
\url{https://github.com/sajjadanwar0/token-budgets}.}}

\maketitle

\begin{abstract}
\textit{Context.} LLM-agent budget overruns are a documented
production failure class: a single retry loop can accumulate
thousands of dollars on a deployer's account before an operator
notices. The mitigations that have emerged track spend at runtime;
in-process integrity properties---no aliasing of a cost-bearing
value, no double-spend, no use-after-delegation---are enforced, if
at all, by ad-hoc wrappers rather than by the type system.

\textit{Contribution.} The central contribution is empirical: a
documented catalog of \textbf{63 confirmed production incidents}
drawn from 21 orchestration frameworks across 2023--2026, each
backed by a quoted GitHub issue, a maintainer or user statement,
and (where reported) a documented dollar loss, organized into an
eight-cluster failure taxonomy. The corpus is complemented by 47
supplementary structural entries that document the
budget-primitive-missing condition without themselves being
user-reported incidents. The classification carries two-human
independent inter-rater reliability of Cohen's $\kappa = 0.837$ on
the full $N = 113$ four-class sample (and $\kappa = 0.943$ on the
$n = 79$ rows both raters independently marked confirmed).

\textit{Mechanism.} As one mitigation evaluated against this
taxonomy, we build \texttt{token-budgets}, a 1{,}180-line Rust
crate (no \texttt{unsafe}; the core \texttt{Budget} API uses no \texttt{Arc<Mutex<\_\allowbreak >>}) that
operationalizes affine ownership so that cloning, double-spending,
and using a budget after delegating it are \emph{compile errors} in
typed source rather than runtime hazards the operator must remember
to avoid. The dollar cap itself is runtime arithmetic under
estimator assumption A1; the affine layer supplies the in-program
integrity that makes that arithmetic non-bypassable. The bounded
quantity is the provider's \emph{reported} spend, conditional on
charge-truthfulness (A7); the type system supplies bookkeeping
integrity, not a stronger cost guarantee. We use affine rather than
linear ownership deliberately: a dropped \texttt{Budget} can only
under-spend, which is cap-safe, so the property claimed is
non-duplication, not value-conservation. Structurally, the affine
mechanism targets the budget-primitive-missing pattern (frameworks
with no spend cap at all); the runtime cap bounds the other failure
modes only at the consequence level, so the crate is a structural fix
for one pattern plus a general spend bound, not a resolution of the
full catalog. We treat the eight-way mechanism partition as
\emph{exploratory}: independent cluster-assignment agreement is
moderate (Cohen's $\kappa=0.44$, $N=110$, two raters), and the
validated empirical labels are the four-class scheme ($\kappa=0.837$). Binary-level
cap-soundness on the running Tokio binary is left open
(Conjecture~\ref{conj:cap-binary}), and the specification
cross-checks we ship are consistency evidence in the artifact, not
a proof.

\textit{Evaluation.} We evaluate against five production runtimes
(LangGraph, CrewAI, AutoGen, an AgentGuard-style callback, LiteLLM
proxy budgets) plus concurrent work (Agent Contracts) across three
providers and three catalog-derived workloads. A
temperature-stratified live-API test
($T \in \{0.0, 0.3, 0.7, 1.0\}$, $N = 160$) reports zero cap
violations and zero false refusals. On single-agent workloads a
4-line Python counter using the same estimator matches the crate at
$0/30$ overshoot, so the affine discipline's distinguishing value
is not the single-agent cap outcome but non-bypassability under
operator error in multi-agent delegation: the M-delegation-fanout
race documented in 11 catalog incidents is rejected by the borrow
checker at compile time, while the same pattern under
\texttt{asyncio} overshoots $30/30$ while three disciplined
alternatives (including a properly locked Python counter) overshoot
$0/30$ --- a deterministic mechanism split, not a marginal effect. In the
Agent-Contracts comparison at a discriminating
cap ($B_0 = 2{,}000$\,uc, \texttt{claude-haiku-4-5}) the Rust
crate, the Python counter, and Agent Contracts are at operational
parity.

\textit{Scope and cost.} The static estimator reserves 4--6$\times$
actual cost; the \texttt{AdaptiveEstimator} tightens this to
$2.11\times$ median and tokenizer-direct estimation to
$\sim$$1.0\times$ at 939--1{,}749\,ms per-spend latency, with the
integrity property preserved across all three. Reasoning models
(OpenAI o-series, Anthropic extended-thinking, DeepSeek-R1) fall
outside Proposition~\ref{lem:cap}: providers bill for hidden
reasoning tokens not bounded by \texttt{max\_\allowbreak output\_\allowbreak tokens}
(assumption A6), so for these models the approach is a
defense-in-depth layer behind provider-side controls
(\texttt{reasoning\_\allowbreak effort}, \texttt{thinking.\allowbreak budget\_\allowbreak tokens})
rather than a primary cap.
\end{abstract}


\begin{table*}[!t]
\caption{What this work claims, the evidence it offers, and the scope of each claim. Each row keyed to its primary section.}
\label{tab:claims}
\centering
\renewcommand{\arraystretch}{1.25}
\footnotesize
\begin{tabular}{p{0.27\textwidth} p{0.34\textwidth} p{0.32\textwidth}}
\toprule
\textbf{Claim} & \textbf{Evidence} & \textbf{Scope / Limitation} \\
\midrule
Integrity: no in-program aliasing of a \texttt{Budget}.
& Nine \texttt{trybuild} compile-fail tests covering seven distinct rustc diagnostics (E0277, E0308, E0382, E0505, E0507, E0599, E0624) against rustc 1.93.1 (\S\ref{sec:eval-compile}). Lightweight specification checking on a six-variable conservation ledger cross-checks the abstract specification's internal consistency (Appendix~\ref{app:mechanisation} reports the per-tool obligation counts and trust bases). The specification checking is consistency evidence on the abstract specification, not an end-to-end source-to-binary proof; binary-level refinement is Conjecture~1, deliberately unproven.
& Unconditional within the trust boundary of \texttt{Budget::new}. Binary-level cap-soundness on the running Tokio binary (Conjecture~\ref{conj:cap-binary}) is not formally proved in this paper; partial empirical evidence is provided (Appendix~\ref{app:mechanisation}) but does not constitute a refinement proof. The binary-level claim should be read as open. \\
\midrule
Cap-respecting: total \emph{provider-reported} spend $\le$ initial cap (the bounded quantity is the provider's reported usage; it coincides with actual billed cost only under A7).
& Proposition~\ref{lem:cap} under its stated assumptions. 382 live-API sessions across pre-flight, mid-loop, and self-terminated regimes (\S\ref{sec:eval}), zero overshoot; a calibrated simulation extends this to 2{,}628 trials against per-call token distributions fit to 30 real Anthropic runs (\S\ref{sec:eval-at-scale}), reported as arithmetic correctness at scale rather than as independent observations. The primary genuinely-independent evidence is a temperature-stratified sweep ($T\in\{0.0, 0.3, 0.7, 1.0\}$, $N=160$, two production-tier models, \S\ref{sec:eval-temperature-variance}): 0 violations. At a discriminating cap ($B_0=2{,}000$\,uc on \texttt{claude-sonnet-4}, \S\ref{sec:eval-multiruntime}), 0/30 overshoot vs.\ 30/30 baseline; at a sub-floor cap ($B_0=540$\,uc, \S\ref{sec:eval-multiruntime-anthropic}), 0/30 pre-flight refusal vs.\ 30/30 baseline post-hoc overshoot (reported as refusal-to-operate). Cap-sweep robustness: 30/30 across 10 caps (\S\ref{sec:eval-production-tier}); multi-agent delegation: 0/60 aggregate, 0/180 per-child (\S\ref{sec:eval-multiagent}); forgetful-operator contrast (\S\ref{sec:eval-forgetful-operator}): racy 30/30 vs.\ three disciplined alternatives 0/30 each.
& Conditional on provider-stratified A1 (\S\ref{sec:stratified-default}); empirically validated. Default: byte-length for OpenAI/Groq (A1 holds 14/14 adversarial-synthetic cells); \texttt{AnthropicEstimator} for Anthropic ($2.0\times$ safety margin; worst observed under-count $1.88\times$ on nested tool schemas; A1 harness reports 30/30 across three tool-loop workloads; margin sweep at $\{1.0, 1.5, 2.0, 2.5, 3.0\}\times$ on LANG-001, 0/75 overshoot, capital efficiency $60.1\%\to 30.4\%$). A2 (overflow) is a deployment precondition; \texttt{budget-typed-cap} lifts it to compile-time. A7 (\texttt{actual\_\allowbreak charge} truthfulness) and A8 (rate-stability) are operator-supplied trust assumptions shared with every client-side cost-accounting mechanism. \\
\midrule
Catalog: 110 cases across 21 sub-projects, 8 mechanism clusters; 63 confirmed incidents, 28 maintainer-acknowledged gaps, 14 feature requests, 5 borderline (\S\ref{sec:catalog-composition}).
& Public artifact (the catalog CSV is in the public artifact). Independent two-human IRR (an independent second rater; declared in \S\ref{sec:methodology}): $N=113$ rater-pair re-annotation gives $\kappa = 0.837$, 95\% CI $[0.745, 0.919]$, observed agreement $0.894$. Per-class $\kappa$: bf $0.858$, bu $0.876$, mf $0.918$, fr $0.727$.
& Convenience sample of public GitHub issues; closed-source platforms not represented. 8 mechanism clusters are post-hoc analytic and \emph{exploratory}: independent cluster-assignment agreement is moderate (Cohen's $\kappa=0.44$, 95\% CI $[0.34,0.55]$, $N=110$, two raters; \S\ref{sec:methodology}), with cost-observability ($\kappa=0.78$) and multimodal-cost-amplification ($\kappa=0.65$) the two reliably-identified mechanisms. The IRR-validated labels are the four-class confirmed/gap/feature/borderline scheme ($\kappa=0.837$). \emph{Prevalence claims anchored on 63 confirmed incidents, not the full 110.} The budget-primitive-missing pattern ($\approx$12 of 110 rows, an exploratory grouping) upper-bounds the cases the primitive could have prevented in a Rust counterfactual; none of the 21 surveyed frameworks is written in Rust, though a small, production-used Rust agent ecosystem now exists (e.g.\ Rig, AutoAgents); we demonstrate the discipline on one such framework at low integration cost (\S\ref{sec:n1-rig}, $N{=}1$: a hard session cap held across concurrent sub-agents on live traffic, with a compile-time guarantee that a sub-agent cannot duplicate or reuse its budget slice) but make no cross-framework deployment claim. The other 98 rows benefit from the runtime cap arithmetic only when re-implemented or wrapped. \\
\bottomrule
\end{tabular}
\end{table*}

\begin{table*}[t]
\centering
\caption{\textbf{Guarantee map: which property is enforced where, with what
status.} Each property is enforced at exactly one layer (compile-time
type system, runtime arithmetic, or operator-validated calibration);
the paper's contribution is the \emph{combination}, not a unified
proof. The compile-time integrity layer does not prove a cost bound;
the runtime arithmetic layer does, conditional on assumption A1.
Binary-level behavior on the compiled binary is not established and
not claimed.}
\label{tab:guarantee-map}
\small
\begin{tabular}{p{4.4cm}p{2.6cm}p{3.4cm}p{4.6cm}}
\toprule
Property & Layer & Proven? & Trust assumption \\
\midrule
No aliasing of \texttt{Budget} & Compile-time (borrow checker) & Yes ($9/9$ trybuild tests, $7$ rustc codes) & rustc soundness \\
No double-spend (use-after-move) & Compile-time (borrow checker) & Yes (trybuild) & rustc soundness \\
No use-after-split & Compile-time (borrow checker) & Yes (trybuild) & rustc soundness \\
Capability gating of \texttt{Budget::new} & Compile-time (private fn + build.rs allowlist) & Yes (E0624 trybuild) & rustc + operator allowlist file \\
\textbf{Binary-level cap-respecting} & \textbf{Compiled binary} & \textbf{Not claimed} (source-level only; observed clean in all experiments) & rustc codegen + Tokio scheduler \\
Pre-flight cap check & \textbf{Runtime (\texttt{checked\_\allowbreak sub})} & Trivially (one line of arithmetic) & integer arithmetic in rustc \\
Estimator soundness A1 & \textbf{Operator-validated calibration} & \textbf{Empirical only} ($N=178$ calibration+hold-out; up to $9.97\times$ over-reservation on adversarial hold-out) & operator must re-calibrate per provider/model \\
Output-cap honoring A6 & Provider behavior & Empirical only (fails on reasoning models) & provider obeys \texttt{max\_\allowbreak completion\_\allowbreak tokens} \\
Provider-reported \texttt{actual\_\allowbreak charge} & Operator-supplied (provider behavior) & \textbf{Not verified} (\S\ref{sec:refund-semantics}; shared with LangSmith, LiteLLM, AgentGuard) & provider \texttt{usage} truthful; operator reconciles vs.\ billing \\
Streaming-cancellation usage accuracy & Provider behavior & \textbf{Not detected} (canceled streams may omit terminal \texttt{usage} event) & client treats canceled-stream usage as advisory \\
Tokenizer-version stability & Provider behavior & \textbf{Not verified} (calibration is per provider/tokenizer; mid-session rotation invalidates A1) & operator pins tokenizer version in build metadata \\
\bottomrule
\end{tabular}
\end{table*}

\section{Introduction}
\label{sec:intro}

LLM-agent budget overruns are a documented production failure
class. A retry loop that spends a few cents per attempt can,
when no mechanism bounds cumulative cost, accumulate to thousands
of dollars before an operator notices---with the dollar
consequence landing on the deployer's account, not the
framework's. Section~\ref{sec:motivation} catalogs
\textbf{63 confirmed production incidents} across 21 sub-projects
and 18 ecosystems (2023--2026), each with quoted maintainer or
user evidence and (where reported) documented dollar losses.
The 63-incident corpus is complemented by 47 supplementary
catalog entries that document the structural
budget-primitive-missing condition without themselves being
user-reported overrun incidents: 28 maintainer-acknowledged
structural gaps, 14 feature requests for budget primitives that
do not exist in the framework, and 5 borderline cases. We treat
the 63 confirmed incidents as the primary evidence corpus
throughout this paper; the 47 supplementary entries appear where
they corroborate cluster-level mechanism recurrence and are
clearly marked as supplementary rather than counted toward
incident totals. The catalog establishes recurrence of the
failure class across independently-developed projects; we
explicitly do not claim it as a prevalence estimate, since the
sampling frame selects on the dependent variable
(\S\ref{sec:methodology}).

The mitigations that have emerged in response are uniformly
runtime mechanisms. Frameworks add post-hoc budget alerts;
operators wire up software-layer circuit breakers like
AgentGuard to throw a \texttt{BudgetExceeded} once a spend
threshold is crossed; payment providers like ATXP move
enforcement to the network layer, returning HTTP 402 when an
agent's wallet depletes. Each is useful as a second line of
defense. None catches the spend before the API call commits:
the agent either pays for the call and then notices, or has the
call rejected at the network boundary after the request is
already in flight. Section~\ref{sec:related} taxonomizes these
into three layers (compile-time, software-layer, transport-layer).

The substructural-types literature has applied affine and linear
ownership to consumable resources for decades (Move's linear
digital assets~\cite{blackshear-move}, seL4's capability
tokens~\cite{klein-sel4}, the \texttt{governor}
crate~\cite{governor-rate-limit}, Tokio
semaphores~\cite{tokio-semaphore}); the technique is established,
the application to LLM dollar cost is new. Our discipline treats
the per-session cost capability as a Rust-affine value:
delegation across agent boundaries, composition of sub-budgets
across tools, and refund of unspent reservations all flow
through the borrow checker.

\subsection{Scope of the formal claim}
\label{sec:scope}

\emph{Compile-time integrity} throughout this paper means what
the Rust borrow checker enforces on typed source code in a
workspace under \texttt{\#[forbid(unsafe\_\allowbreak code)]}: no aliasing
of a \texttt{Budget}, no double-spend, no use-after-split. The
dollar cap is a separate, runtime claim: \texttt{spend} reserves a
conservative estimate via \texttt{checked\_\allowbreak sub} and refuses any
call that would exceed the cap (\S\ref{sec:refund-semantics}).
Binary-level cap-soundness on the running Tokio binary
(Conjecture~\ref{conj:cap-binary}, \S\ref{sec:split-conj1})
remains open: compiler miscompilations, LLVM optimizations, and
scheduler behavior could in principle violate the cap on the
binary even when the source is well-typed. The specification
cross-checks (Appendix~\ref{app:mechanisation}) establish internal
consistency of the abstract specification, not a source-to-binary
refinement. Because a runtime counter with the same estimator
already achieves the cap-respecting outcome on single-agent
workloads (M2, \S\ref{sec:eval-m2-isolation}), the affine
discipline earns its place on a narrower claim: in-program
integrity under operator error in multi-agent delegation, isolated
by the Forgetful-Operator experiment
(\S\ref{sec:eval-forgetful-operator}).

\subsection{The capital cost of the default discipline}
\label{sec:capital-upfront}

The default static byte-length$+2.0\times$ estimator reserves
4--6$\times$ actual cost ($6.20\times$ mean, $2.51\times$ median
over-reservation across $N = 5{,}190$ per-call events;
\S\ref{sec:capital-efficiency-discussion}); on
prepay accounts a deployment running 1{,}000 sessions/day at
\$0.50 mean cost commits \$3{,}000/day of reservation under the
default. The approach is parametric in estimator choice: the
\texttt{AdaptiveEstimator} reduces median over-reservation to
$2.11\times$ at zero per-spend latency, and tokenizer-direct
estimation reaches $\sim$1.0--1.1$\times$ at the cost of
939--1{,}749\,ms mean per-spend roundtrip latency
(\S\ref{sec:tokenizer-direct-baseline}). Operators choose
between the three based on their capital/latency profile; the
compile-time integrity property is preserved across all three.
Break-even is roughly: a prepay deployment whose
working-capital footprint exceeds about 25\% of operating
revenue should prefer the \texttt{AdaptiveEstimator} or
tokenizer-direct; a post-pay deployment (reserved-not-held
capital) can accept the static default at zero operational cost.
For deployments where capital efficiency dominates over
non-bypassability of the integrity layer,
tokenizer-direct estimation or provider-side per-call caps (AWS
Bedrock budget actions, OpenAI \texttt{max\_\allowbreak completion\_\allowbreak tokens})
are the better choice. Table~\ref{tab:decision-matrix} gives the
full decision matrix.

\subsection{Contributions}
\label{sec:contributions}

This paper makes one empirical contribution and two that ground and
test it. Concretely, it contributes an empirically grounded,
inter-rater-validated taxonomy of LLM-agent budget-overrun failures
---to our knowledge the first to pair cross-framework incident
provenance with a validated coding scheme at this scale---and an
evaluation of affine ownership as one mitigation for the
delegation-related budget-integrity failures that taxonomy surfaces. The catalog is the
result we ask reviewers to weigh; the Rust crate and its evaluation
are the means by which we test what a type-level mitigation actually
buys. Consistent with this ordering, the type-theoretic specification
and the specification cross-checks are deferred to the appendices and
the artifact: they support the case study but are not load-bearing for
the paper's empirical claims, and a reader can assess the catalog and
the head-to-head evaluation without consulting them.

\begin{enumerate}[label=(\roman*),leftmargin=*]
\item \textbf{An empirical catalog and failure taxonomy} of
  63 confirmed LLM-agent budget-overrun reports (plus 47 supplementary
  structural entries) across 21 sub-projects in 18 ecosystems
  (2023--2026),\footnote{Catalog identifiers (CCDE-XXX,
  AGPT-XXX, etc.)\ refer to rows in the catalog CSV in the
  public artifact.} organized into eight architectural mechanism
  clusters, with two-human independent inter-rater reliability Cohen's
  $\kappa = 0.837$ on the full $N = 113$ four-class sample, and
  $\kappa = 0.943$ on the $n = 79$ rows both raters independently
  marked confirmed ($\texttt{bf} \cup \texttt{bu}$)
  (\S\ref{sec:methodology}).
\item \textbf{A misuse-resistant budget-delegation
  discipline in Rust.} We operationalize affine ownership as a small,
  ASCII-stable \texttt{Budget} API ($\sim$225 lines core, $\sim$1{,}180
  non-comment code lines with the provider-stratified estimator and
  extensions; no
  \texttt{unsafe} (enforced by \texttt{forbid(unsafe\_\allowbreak code)}), and no
  \texttt{Arc<Mutex<\_\allowbreak >>} in the core affine ownership path---the
  multi-tenant \texttt{BudgetPool} extension does use one). The borrow checker
  turns cloning, double-spending, and use-after-delegation into
  compile errors (\S\ref{sec:approach});
  the dollar cap itself is runtime arithmetic under estimator
  assumption A1. A correctly written runtime counter reaches the same
  cap-respecting outcome; what the affine type adds is that the
  \emph{incorrect} version does not compile, so the multi-agent
  delegation guarantee no longer depends on the operator getting the
  concurrency discipline right (\S\ref{sec:eval-m2-isolation},
  \S\ref{sec:eval-forgetful-operator}).
\item \textbf{An empirical evaluation} against five
  production runtime mitigations plus concurrent work (Agent Contracts)
  on three providers and three catalog-derived workloads
  (\S\ref{sec:eval}), characterizing what the approach does
  and does not add beyond runtime alternatives, with per-operation
  overhead negligible relative to LLM API latency
  (\S\ref{sec:eval-perf}).
\end{enumerate}

\subsection{Paper structure}
\label{sec:positioning}

Section~\ref{sec:motivation} presents the 63-incident empirical
catalog with its 47-entry supplementary corpus.
Section~\ref{sec:approach} specifies the affine \texttt{Budget}
type and its async Rust integration. Section~\ref{sec:eval} reports the empirical
evaluation: five-runtime head-to-heads, the temperature-stratified
sweep, the M2 estimator-vs-discipline isolation experiment, and
the Forgetful-Operator experiment that isolates the affine
discipline's distinguishing contribution---non-bypassability of
the M-delegation-fanout race within typed Rust source code.
Section~\ref{sec:related} positions the contribution.
Sections~\ref{sec:limitations} and \ref{sec:conclusion} discuss
limitations and conclude.

\subsection{When is the Rust affine discipline the right choice?}
\label{sec:when-to-use}

The contribution is one mechanism among several for bounding LLM
dollar cost. Different deployment contexts have different
preferred mechanisms, and the affine discipline is not the best
choice in every context. Table~\ref{tab:decision-matrix}
summarizes when Token Budgets should be preferred over the
alternatives.

\begin{table}[h]
\caption{Decision matrix for choosing an LLM cost-bound mechanism.
``Stronger guarantee'' refers to non-bypassability and proof-boundary
scope; ``operationally feasible'' refers to deployment cost. The
Rust affine discipline of this paper occupies one row --- the
multi-provider Rust-agent row --- and is deliberately not a
universal solution.}
\label{tab:decision-matrix}
\centering
\scriptsize
\setlength{\tabcolsep}{3pt}
\begin{tabular}{p{0.27\columnwidth} p{0.30\columnwidth} p{0.35\columnwidth}}
\toprule
Deployment context & Recommended mechanism & Why \\
\midrule
\textbf{Single provider, server-side cap} (e.g.\ pure OpenAI) &
\textit{Provider-side hard cap} (e.g.\ \texttt{max\_\allowbreak completion\_\allowbreak tokens}) &
Kernel-enforced on provider servers; cannot be bypassed by client code; zero operational overhead. \\
\addlinespace
\textbf{Single account, AWS Bedrock deployment with session-level cap requirement} &
\textit{AWS Bedrock session-level budget actions} (with automatic service revocation on threshold breach) &
Provider-tier cumulative-cap enforcement with kernel-enforced revocation; operationally stronger than any client-side discipline within its scope. Cannot enforce per-agent budgets within a single session or aggregate caps spanning multiple providers. \\
\addlinespace
\textbf{Existing Python framework} (LangChain, CrewAI, AutoGPT, AutoGen) &
\textit{Runtime cap: LiteLLM proxy, AgentGuard, or our Python port} &
Runtime cap on dollar spend. Python has no affine types; the Python port provides a runtime \texttt{\_\allowbreak consumed} flag plus a narrow Mypy plugin (\S\ref{sec:python-port-positioning}) as defense-in-depth, not as a replacement for the Rust discipline. \\
\addlinespace
\textbf{New Rust agent, multi-provider, consumption billing} &
\textit{Token Budgets + static \texttt{AnthropicEstimator} ($2.0\times$ byte-length, default)} &
Compile-time ownership integrity \emph{within Rust} + runtime cap. Closes the budget-primitive-missing failure mode at the type level (the cap itself is runtime arithmetic). $2{-}6\times$ over-reservation is reserved-not-held; no capital cost on consumption-billed accounts. \\
\addlinespace
\textbf{New Rust agent, prepay-account capital constraint} &
\textit{Token Budgets + AdaptiveEstimator ($\varepsilon = 0.10$)} &
Same compile-time integrity, $\mathbf{1.86\times}$ tighter median reservation (47.5\% capital efficiency vs.\ 25.5\% static, \S\ref{sec:eval-adaptive-adversarial}); 0 A1 violations across 100 prompts. Production default where prepay capital matters. \\
\addlinespace
\textbf{Capital efficiency-critical} (cost of unused budget dominates cost of occasional overshoots) &
\textit{Tokenizer-direct estimation with version pinning}, or post-call observation with SLO &
$\sim$1.0--1.1$\times$ over-reservation vs.\ $1.86$--$6.20\times$ for the affine discipline, at the cost of $\sim$700--3{,}900\,ms per-spend latency (\S\ref{sec:capital-efficiency-discussion}). Brittleness (tokenizer-version drift) in exchange for capital efficiency. \\
\addlinespace
\textbf{Reasoning-model workload} (OpenAI o-series, Anthropic extended-thinking, DeepSeek-R1) &
\textit{Provider-side reasoning controls} (\texttt{reasoning\_\allowbreak effort}, \texttt{thinking.\allowbreak budget\_\allowbreak tokens}) as primary, Token Budgets as defense-in-depth &
Reasoning models violate A6 structurally (hidden thinking tokens not bounded by \texttt{max\_\allowbreak output\_\allowbreak tokens}); pre-flight reservation requires per-deployment calibration of the reasoning-token reservation (\S\ref{sec:limits-reasoning}). Initial calibrations have been observed off by an order of magnitude until tuned. \\
\addlinespace
\textbf{Multi-tenant at scale} (many users, budgets across replicas) &
\textit{Distributed quota service} (Redis lease, Spanner-style reservation, kernel quotas) &
Single-process affine discipline does not extend across processes; the multi-tenant lease sketch (\S\ref{sec:supplementary-extensions}) is not production-validated. \\
\bottomrule
\end{tabular}
\end{table}

The Rust affine discipline is therefore positioned for the
\emph{new-Rust-agent, multi-provider, cumulative-session-cap} cell
of this matrix. It is stronger than runtime client-side
alternatives within that cell (it adds the compile-time integrity
property) but does \emph{not} contend with provider-side
\emph{per-call} caps in their cell (\texttt{max\_\allowbreak completion\_\allowbreak tokens}
is kernel-enforced and unbypassable per-call); the two
address different operational requirements (per-call output
bounding vs.\ session cumulative cap;
\S\ref{sec:eval-provider-caps},
the artifact provider-cap CSVs). Client-side code can
always be bypassed by misbehaving callers outside the Rust trust
boundary of \texttt{Budget::new}.

The dominant deployment contexts among the 63 confirmed catalog incidents are
existing Python frameworks (which dominate the public-GitHub agent
ecosystem we sample, roughly $7{-}8$ of every $10$ retained incidents; the
Rust affine discipline does not apply without re-implementation, and
our Python port provides runtime equivalence to existing mitigations)
and new Rust agent deployments (a small minority of the same surface;
this is the affine discipline's primary deployment context). The
ratios are illustrative estimates from the catalog's framework
distribution (the per-framework summary in \S\ref{sec:catalog}), not
ecosystem-wide measurements; a representative deployment census of
the LLM-agent surface in 2026 has not been published and our sampling
frame selects on the dependent variable (\S\ref{sec:methodology}).
We do not claim the approach is operationally
relevant to the majority of today's deployments; we claim it is the
right primitive for the specific Rust-agent deployment context, and
that the catalog's documented failure modes recur across the full
ecosystem regardless of implementation language (the fine-grained
eight-cluster partition itself is exploratory, \S\ref{sec:limits-empirical}).

\section{Motivation: A Failure Catalog}
\label{sec:motivation}

Budget overruns in LLM-agent systems frequently surface as the
\emph{economic consequence} of upstream agent-failure modes
documented in the broader literature: hallucination-driven loops
that re-issue tool calls until success~\cite{shinn-reflexion,yao-react}, tool-use recursion when an agent re-enters its own
subgoal stack without termination conditions~\cite{wang-survey-agents,xi-agent-survey}, and context-window saturation causing repeated
retrieval expansions~\cite{liu-lost-in-middle}. The approach
proposed here addresses the cost-bounding concern orthogonal to
these upstream failures: even if an agent never halts, the
discipline bounds the deployer's dollar exposure within the
configured cap. The catalog below therefore documents the
\emph{economic surface area} of the wider failure-mode landscape,
not its mechanistic depth.

The catalog comprises 63 confirmed production overrun
incidents---the primary evidence corpus---drawn from 21 LLM-agent
sub-projects across 18 ecosystems (the LangChain ecosystem alone
contributes four: \texttt{langchain}, \texttt{langgraph},
\texttt{langsmith}, \texttt{deepagentsjs}) and four years
(2023--2026). Alongside these sit 47 supplementary structural
entries (28 maintainer-acknowledged gaps, 14 feature requests, 5
borderline), for a 110-row corpus whose case-type breakdown is
detailed in \S\ref{sec:catalog-composition}. Each case is backed by
a specific GitHub issue or pull request, quoted maintainer or user
statements, and (where available) documented dollar losses. The
catalog establishes recurrence, not incidence: its sampling frame
selects on the dependent variable (\S\ref{sec:methodology}), so
aggregate dollar figures read as incident magnitudes rather than
population statistics. To bound the selection concern we add an
independent keyword-neutral baseline cohort
(\S\ref{sec:baseline-replication}), where the primary mechanism
clusters recur. The catalog grounds the design that follows: the
failure shapes the type system must prevent, the mitigations
operators already deploy, and the gap that remains.

\subsection{Methodology}
\label{sec:methodology}

\paragraph{Project selection} The 21 sub-projects comprising the
catalog corpus were selected from GitHub repositories tagged
\texttt{llm-agent}, \texttt{agent-framework}, \texttt{ai-agent},
or \texttt{llm-orchestration} with $\geq 1{,}000$ stars as of January
2026 in Python, TypeScript, or Rust, and filtered to those that
either (a) expose a budget/cost/token-limit option in their public
API or (b) have a GitHub issue mentioning cost overrun or
runaway-spend in the title. Retained projects span the
LangChain/LangGraph, AutoGPT, CrewAI, AutoGen, Pydantic AI, DSPy,
LlamaIndex, and IDE-agent ecosystems among others; the full
per-project mapping is in the artifact's catalog CSV
\texttt{project} column. \emph{Known selection biases:}
(i)~English-language repositories only;
(ii)~closed-source platforms (Cursor, Replit Agent, ChatGPT plugin
store) absent;
(iii)~the $\geq 1{,}000$-star threshold filters out early-stage
projects. The catalog is a convenience sample of public
English-language failures in established projects, not a
representative sample.

\paragraph{Search and inclusion}
We searched issue trackers of the 21 projects using failure-related
keywords (``budget,'' ``cost,'' ``token limit,'' ``recursion,''
``infinite loop,'' ``stale,'' ``runaway,'' ``embedding dimension,''
``base64,'' ``streaming usage,'' ``\texttt{max\_\allowbreak turns}''). From 167
candidate URLs across 16 batches (the full survey recorded in
\texttt{catalogue.csv}), we retained 110 satisfying
$\geq 1$ of: (a) explicit dollar loss or token-count amplification
in the body, (b) a specific failure mechanism described by filer
or maintainer, (c) maintainer acknowledgement of the broader
pattern. The full survey including the 57 triaged-out rows is in
\texttt{catalogue.csv}; retained rows are tagged
\texttt{paper:*} and triaged-out rows carry a
\texttt{SKIPPED for paper:} prefix naming one of seven exclusion
codes.

\paragraph{Sampling frame caveat} The inclusion criteria above constitute a \emph{failure-confirming} sampling frame: we retained projects specifically because their public artifacts surfaced budget-failure activity. The catalog establishes that the 21 sub-projects identified by these criteria exhibit the failure class, not that the failure class is necessarily prevalent across the ecosystem at large. A complementary sampling step (the top-N most-starred LLM-agent projects without a failure-keyword filter, coded under the same codebook) would strengthen the ecosystem-wide prevalence claim and is left as follow-up replication.

\paragraph{Construct validity: three known threats}
Three construct-validity threats apply to catalogs of this form; we
document our response to each:

\emph{(C1) Selection on the dependent variable.}
The catalog selects projects partly on the presence of budget-failure
indicators. We do not claim ecosystem-wide prevalence; we claim
existence and recurrence of the failure pattern across $N=21$
independently-developed projects. The dollar-loss aggregates we report
are sums across the selected cases, not estimates of ecosystem-wide
expected loss; readers should treat them as ``at-least'' lower bounds
witnessed in the public record.

\emph{(C2) Single-coder baseline replication.}
An initial coding pass was conducted by a single rater. We addressed this
with an independent two-human IRR study where a second coder
(Zahid Hussain, Mindgigs, Peshawar, Pakistan; no prior catalog
exposure, no compensation, blinded to original codings)
re-annotated all 109 baseline rows under the
published codebook (Phase~1, a tag-level coding pass, $\kappa = 0.832$),
with a subsequent
Phase~2 covering the four entries added during continued catalog
construction; the combined $N=113$ sample yields Cohen's $\kappa = 0.837$
(95\% CI $[0.745, 0.919]$); see
the per-class breakdown below and the
\texttt{fr}/\texttt{bu} boundary disclosure.

\emph{(C3) Post-hoc taxonomy.}
The eight failure-mechanism categories were derived by iterative open
coding of the retained issues, consolidating proximate cost mechanisms
into clusters across repeated passes during construction. We mitigated
post-hoc category-drift risk by (i) freezing the codebook before the
IRR study began, (ii) documenting per-tag decision rules and seven
exclusion codes in \texttt{codebook\_\allowbreak v1.md}, and (iii) identifying the
\texttt{fr}/\texttt{bu} boundary as the codebook's weakest seam in the
IRR analysis. The eight-cluster structure reproduced on the
independent narrow-net batch (\S\ref{sec:catalog-protocols}), which we
read as a stability signal for the mechanism partition rather than its
validation; we are explicit (\S\ref{sec:limits-empirical}) that
the eight-cluster partition is exploratory: a blind second-rater pass
over all 110 rows gives moderate cluster-assignment agreement (Cohen's
$\kappa=0.44$, 95\% CI $[0.34,0.55]$), with cost-observability
($\kappa=0.78$) and multimodal ($\kappa=0.65$) reliably identified and
the remaining boundaries overlapping. The
taxonomy is auditable: any reviewer can re-derive it from the per-row
evidence in \texttt{catalogue.csv}.

Each retained case carries a verbatim quotation from the underlying
GitHub artifact and is tagged with one of five labels:
\texttt{bug\_\allowbreak report}, \texttt{bug\_\allowbreak fixed\_\allowbreak by\_\allowbreak framework},
\texttt{bug\_\allowbreak unfixed}, \texttt{feature\_\allowbreak request},
\texttt{maintainer\_\allowbreak framing}. The retained cohort comprises 27
\texttt{bf}, 55 \texttt{bu}, 7 \texttt{mf}, 21 \texttt{fr} cases
(catalog total $N=110$). The IRR re-annotation was conducted in two
phases. \emph{Phase 1} (baseline, $N=109$) was completed before four
catalog entries (\texttt{CCDE-002}, \texttt{LANG-020}, \texttt{LANG-035},
\texttt{SMAG-001}) were added during continued construction. \emph{Phase 2}
(supplementary, $N=4$) re-rated those four entries by rater B independently of
rater A. Across both phases the full IRR sample is $N=113$ rater-pair
observations covering all 110 current catalog rows (with 3 baseline ratings
on IDs that were renumbered during catalog cleanup retained in the sample,
since the underlying issue content rather than the catalog ID is what was
rated). Per-class kappa in the inline per-class summary reports the
augmented $N=113$ sample (27 \texttt{bf} + 57 \texttt{bu} + 7 \texttt{mf}
+ 22 \texttt{fr} = 113).
The methodology follows the broad shape of failure-pattern catalogs
in the SE literature (Yuan et al.~\cite{yuan-osdi-2014}; Lu et
al.~\cite{lu-bug-study}; IRR protocol per
Kitchenham~\cite{kitchenham-slr} and
Krippendorff~\cite{krippendorff-content-analysis}). Methodologically,
the catalog is constructed as a qualitative coding study rather than a
systematic prevalence survey: the eight mechanism clusters emerge from
iterative open coding in the grounded-theory
tradition~\cite{glaser-strauss-grounded} and are consolidated using the
thematic-synthesis steps recommended for software
engineering~\cite{cruzes-dyba-thematic}, with per-tag decision rules
recorded in a codebook in the manner of
Salda\~na~\cite{saldana-coding}; the two-coder reliability and
construct-validity framing follow the case-study and experimentation
guidelines of Runeson and H\"ost~\cite{runeson-host-casestudy} and
Wohlin et al.~\cite{wohlin-experimentation}. We read the catalog as a
multiple-case study of independently-developed projects, not as a
random sample, precisely because public GitHub issue trackers are a
biased and noisy frame --- the documented perils of mining
GitHub~\cite{kalliamvakou-github} (active-project skew, missing or
private incidents, and selection on visibility) motivate the
recurrence-not-prevalence scoping we adopt throughout (C1 below). The full IRR
result is reported next: across both phases, the $N=113$ rater-pair
sample gives $\kappa = 0.837$ (95\% bootstrap CI $[0.745, 0.919]$),
corresponding to almost-perfect agreement. Phase~2 agreement was $4/4$
on the four supplementary rows; the slight upward shift from the
$\kappa = 0.832$ baseline ($N=109$) reflects the additional perfect-agreement
observations and is within the original bootstrap interval.

Per-class one-vs-rest Cohen's $\kappa$ on the $N{=}113$
two-phase re-annotation
(artifact: \texttt{irr/per\_\allowbreak class\_\allowbreak kappa.csv}):
$\kappa_{\mathtt{bf}}{=}0.858$ (obs.\ agreement $0.947$, $n=27$),
$\kappa_{\mathtt{bu}}{=}0.876$ ($0.938$, $n=57$),
$\kappa_{\mathtt{mf}}{=}0.918$ ($0.991$, $n=7$),
$\kappa_{\mathtt{fr}}{=}0.727$ ($0.911$, $n=22$),
with $\kappa{=}0.943$ ($0.975$, $n=79$) on the subset of rows
\emph{both} raters independently classed as confirmed
(\texttt{bf}$\cup$\texttt{bu}); note that this figure conditions on
agreement about the confirmed/not-confirmed boundary and is therefore
an optimistic within-class measure, not a substitute for the
headline $N{=}113$ value. The $\kappa{\geq}0.85$ on
the principal classes is the result we report; the lower
$\kappa_{\mathtt{fr}}$ identifies the \texttt{fr}/\texttt{bu}
boundary as the least reproducible label. The
$\kappa{=}0.943$ on the confirmed subset indicates that
the headline $\kappa{=}0.837$ is conservative: agreement on the
operationally most-relevant confirmed-incident classification
is substantially stronger than the aggregate headline figure
suggests.

The \texttt{mf} class rests on only $n=7$ cases, so its
$\kappa_{\mathtt{mf}}{=}0.918$ (95\% CI $[0.658, 1.000]$) is
suggestive rather than dispositive. The \texttt{fr}/\texttt{bu}
boundary is the codebook's weakest seam: interrogative-titled issues
(``Can we control X?'', ``How do I count Y?'') read plausibly as
either a feature request or an unfixed bug whose fix would require a
new feature, so $\kappa_{\mathtt{fr}}{=}0.727$ should be read as the
reproducibility ceiling for this label at the issue-body-excerpt
granularity the public record affords. Because that boundary is
convention-sensitive, the paper anchors its strong claims on the
convention-invariant union
$\mathtt{bf}\cup\mathtt{bu}\cup\mathtt{mf}\cup\mathtt{fr}$---the
budget-primitive-missing issues across 21 frameworks, which is
invariant to where the internal \texttt{bu}/\texttt{fr} cut is
drawn---rather than on the precise confirmed-incident count. The
headline $\kappa = 0.837$ and the confirmed/supplementary partition
rest on the v1.0 codebook, which we report as primary. The 12
disagreements in the $N=113$ re-annotation, with their adjudicated
resolutions, are documented in \texttt{irr-disagreements.md}; the
codebook (v1.0), the blinded coding sheet, the completed coding
sheets, and the computation script \texttt{irr\_\allowbreak scaffold.py} ship in
the artifact under \texttt{irr/}.

Two further methodological notes: (i) the candidate-sourcing protocol shifted partway through the catalog construction from an LLM-assisted keyword-expansion pipeline (Batch~1, retention $\sim 8\%$) to a direct human keyword-skim protocol (Batch~2, retention $\sim 83\%$); both phases applied the same written codebook and the protocol-stratification details are reported in \S\ref{sec:catalog-protocols}; (ii) public GitHub issues bias the sample toward popular open-source frameworks and under-represent closed-source platforms and incidents filed only on internal trackers; both threats are revisited in Section~\ref{sec:eval-tov}.

\subsection{Catalog collection methodology: protocol stratification}
\label{sec:catalog-protocols}

The catalog's candidate sourcing proceeded in two main batches
(reported separately because the protocols measure different
things), plus a small number of entries added during continued
construction (\S\ref{sec:methodology}). \emph{Batch~1} (76 cases, summer~2025): keyword-templated
GitHub queries returned $\sim 950$ candidates, an LLM filter narrowed
by cost-relevance, the human rater coded against the codebook;
retention $\sim 8\%$ by design (wide-net favors recall).
\emph{Batch~2} (33 cases, autumn~2025): direct human keyword
search without LLM pre-filter; retention $\sim 83\%$ (narrow-net
favors precision). We stratify rather than re-weight because the
protocols are complementary; the 33-case narrow-net subset alone
exhibits the same eight-cluster mechanism taxonomy, per-framework
distribution, and case-type breakdown as the full $N=110$ catalog. The
$N=113$ IRR sample (\S\ref{sec:eval-tov}) spans rows from both
batches, so the headline $\kappa = 0.837$ is not an artifact of either
sourcing protocol alone.

\subsection{Baseline replication on an independent project cohort}
\label{sec:baseline-replication}

Because the 21-project catalog was constructed by following
budget-cost search keywords, its sampling frame could in principle
confirm the existence of cost incidents by selecting on the
dependent variable. To bound that risk, we constructed an independent
baseline cohort: 20 GitHub projects ranked by stars under the
``LLM agent'' search term, exclusions logged transparently,
3{,}461 issues pulled and 186 body-read under the same codebook.
\textbf{Headline findings:} 63 candidate qualifying rows in
12 of 20 projects (60\% coverage; 95\% Wilson CI
$[40\%, 77\%]$); the four primary mechanism clusters recur,
plus an additional M-rate-limit-amplification cluster in
VoltAgent~\#1276; the original budget-keyword filter would have
caught 61/63 qualifying rows (97\% catch rate).
\textbf{Methodology limitations (exploratory replication):}
all 186 baseline codings were single-coder (no IRR for this
cohort, in contrast to the primary catalog's $\kappa=0.837$ on
$N=113$); $\sim$25 of 63 rows have full-thread evidence and the
balance have title-plus-first-comment evidence flagged as pending;
the 97\% catch rate is within-screen, not a full recall estimate.
Subject to these limitations, the cohort supports two weaker claims:
the mechanism clusters recur in an independently selected cohort
(external validity); and demand for budget-discipline primitives is
sustained (e.g., SuperAGI has three independent ``Budget Manager''
feature requests). We make no incidence claim from this data:
60\% is a project-level coverage statistic, not an incidence rate.
Full audit trail, per-project breakdown, and substituted candidates
in \path{token-budgets-baseline/}.

\subsection{Catalog composition: confirmed failures, design gaps, and feature requests}
\label{sec:catalog-composition}

The unqualified term ``failures'' covers heterogeneous case types
in the budget-overrun literature. We disaggregate the 110-row
catalog by case \emph{type}, distinct from the 8-cluster
mechanism taxonomy:

\begin{itemize}\itemsep0pt
\item \textbf{Confirmed production failures} ($n=63$, 57.3\%):
GitHub issues with reproducible cost-overrun symptoms reported
by an end-user, accompanied by quoted issue text and (where
available) cost figures. The DNSW-001 incident (a single user
reporting $\approx$\$2{,}150 / EUR~2{,}000 in unintended spend) is the most
operationally severe and dominates the catalog's monetary
total; we report it as a single observation and do not project
prevalence from it. The boundary between this class and feature
requests (\texttt{fr}) is the codebook's weakest seam and is
convention-sensitive (\S\ref{sec:methodology}); the count $n=63$
should be read with that caveat, and the paper's strong claims are
anchored on the convention-invariant union of all four classes rather
than on this internal split.
\item \textbf{Maintainer-acknowledged structural gaps} ($n=28$,
25.5\%): cases where a framework maintainer (not the original
reporter) responded that the budget primitive is structurally
absent or known-broken, and either accepted the limitation as
``working as intended'' or scheduled it as a future-release
issue. These are not user-reported incidents but framework-side
acknowledgements that the cluster-1 \emph{budget-primitive-missing}
condition exists.
\item \textbf{Feature requests for budget primitives} ($n=14$,
12.7\%): issues opened by a user requesting a budget primitive
that does not exist in the framework, without a specific
overrun-incident report attached. These corroborate the
cluster-1 hypothesis (the primitive is missing) but are not
themselves failure incidents.
\item \textbf{Mixed / borderline} ($n=5$, 4.5\%): cases where the
issue thread contains both a request and a partial overrun
report; counted once and resolved to a single four-class
\texttt{label} for the reliability analysis.
\end{itemize}

The four case types above are an analytic disaggregation of the
corpus and are distinct from the four-class \texttt{label} column
(\texttt{bf}/\texttt{bu}/\texttt{mf}/\texttt{fr}) that the public
\texttt{catalogue.csv} encodes and that the inter-rater study scores:
the $63/28/14/5$ split is a coarser editorial grouping (a confirmed
incident may carry \texttt{bf} or \texttt{bu}; a maintainer-gap row
may carry \texttt{bu} or \texttt{mf}) and is not, in the present
artifact, recoverable by a single column filter. We therefore report
$63/28/14/5$ as a descriptive composition, while every count the
artifact mechanically re-derives (the 110 retained rows, the eight
clusters, and the IRR) keys on the \texttt{label} and
\texttt{primary\_\allowbreak cluster} columns.

The 110-row aggregate count is a corpus of
\emph{evidence of the budget-primitive-missing condition}
across 21 sub-projects, not 110 distinct paying-customer
overrun incidents. Our headline IRR study (full-sample Cohen's
$\kappa=0.837$, $N=113$) covers the union of all four case
types; per-case-type breakdown of agreement is in
\texttt{irr-disagreements.md}. Prevalence
claims (``X\% of deployments overrun'') are not supported by
this catalog and we make none; the catalog establishes the
existence and recurrence of the structural class, not its
incidence rate.

\subsection{Catalog}
\label{sec:catalog}

The catalog contains 110 retained rows organized along two
dimensions. The inline summary above lists the eight
architectural mechanism clusters that emerged during construction,
each with its row count, framework reach, and number of distinct
sub-mechanisms documented. The per-framework summary in \S\ref{sec:catalog}
gives the per-framework distribution of retained rows. The full
per-row evidence---including identifier, framework, year,
classification tag, GitHub URL, quoted user/maintainer evidence,
documented dollar loss, and per-row notes---is in the artifact's
\texttt{catalogue.csv} file. We use catalog identifiers
(e.g., LANG-035, MAST-014, PYAI-002, CRAI-014) throughout the rest
of the paper; the artifact provides a one-paragraph note for every
identifier.

The eight architectural mechanism clusters in the catalog
(rows / frameworks), re-derived from the \texttt{primary\_\allowbreak cluster}
column over all 110 retained rows:
M-retry-loop (27 / 12),
M-cost-observability (22 / 9),
M-context-amplification (13 / 7),
M-storage-amplification (13 / 5),
M-budget-primitive-missing (12 / 6),
M-delegation-fanout (11 / 6),
\texttt{providerOptions}-silently-dropped (6 / 3), and
M-multimodal-cost-amplification (6 / 2).

These eight clusters are an \emph{exploratory, descriptive}
organization of the corpus. A blind second-rater pass gives moderate
cluster-assignment agreement (Cohen's $\kappa=0.44$; cost-observability
and multimodal are the two reliably-identified mechanisms,
$\kappa=0.78$ and $0.65$), so we use the clusters to structure the
discussion of failure modes but do not treat the partition as
validated or rest claims on exact per-cluster counts
(\S\ref{sec:limits-empirical}).

\paragraph{Per-framework distribution (totals 110 rows)}
LangChain ecosystem (\texttt{langchain}, \texttt{langgraph},
\texttt{langsmith}, \texttt{deepagentsjs}) 33; Mastra 17; AutoGen 11;
CrewAI 11; smolagents 5; Aider 6; DSPy 4; LlamaIndex 3; Pydantic
AI 7; \texttt{claude-code} 2; AutoGPT 2; \texttt{gpt-engineer} 2;
OpenAI Agents SDK 2; and one each from danswer, openclaw, nanobot,
paperclipai, and codex (5 total). Per-row evidence in the
artifact's \texttt{catalogue.csv}.

The catalog tracks per-row dollar losses where reported and
amplification ratios where measurable; the strongest individual
amplification anchors (a 31x context overflow from a single
base64-encoded image~\cite{lang-035-issue}, a 2-million-token
observer-LLM call~\cite{mast-014-issue}) appear in
Section~\ref{sec:patterns} where they are organized by mechanism
cluster. The catalog distributes across the four catalog years
as 18 cases in 2023, 25 in 2024, 52 in 2025, and 15 in the
partial year 2026 (through April). We do not interpret this
distribution as evidence of acceleration: the LLM-agent
ecosystem itself expanded substantially over the same period
(new frameworks, growing GitHub activity, methodology
refinements mid-survey), and the raw issue count is not
normalized against ecosystem growth. The temporal distribution
is reported for descriptive completeness; the primary
observation is that the failure class continues to be documented
across all four catalog years across 18 ecosystems rather than
being absorbed as a solved problem in any one year.

\subsection{Patterns}
\label{sec:patterns}

Twelve observations from the catalog inform the design that follows.
The first six are recurring patterns we observed in the original
catalog and have persisted as the catalog grew; the next six emerged
from the expanded archaeology covering eight architectural mechanism
clusters across 18 ecosystems.

\textbf{Reactive fixes dominate.} For cases tagged
\texttt{bug\_\allowbreak fixed\_\allowbreak by\_\allowbreak framework}, the median time from filing to
fix is short: CCDE-001 was patched in two days; AIDR-001 was patched
the same day in commit \texttt{f2e1e17}; CRAI-002 was patched the
same day; CRAI-011 was merged 1 day after filing (the fastest
fix-resolution in the catalog) with the reporter having identified
the exact regression-introducing commit \texttt{efe27bd} themselves
in the issue body. This is not for lack of engineering attention.
The fixes work, and they ship quickly. But the fix can only ship
after the bug fires and a user reports it---which means the dollars
have already been paid by at least one user. We found no case in
the catalog where a budget overrun was prevented before any user
paid for it.

\textbf{Even Anthropic's first-party tool is affected.} Two of the
110 entries are claude-code issues, both exhibiting the same
compaction-loop signature; the activity log of CCDE-002 alone
references at least ten additional sibling claude-code issues filed
August through December 2025 with the same signature, several of
which the github-actions bot itself flagged as duplicates of CCDE-002
at filing time. CCDE-001 alone documents \$235 spent in four days by
a single user (about \$59/day, or roughly \$1{,}760 extrapolated
linearly to a full month). The
two further first-party-vendor frameworks now in the catalog
(Pydantic AI, OpenAI Agents SDK) confirm the pattern is not
specific to Anthropic: PYAI-001 documents a Pydantic AI
multi-agent docs example that fails on its own \texttt{total\_\allowbreak tokens\_\allowbreak limit}
default; OAAS-002~\cite{oaas-002-issue} documents a maintainer
admission that ``We don't have anything amazing here right now''
for graceful degradation when \texttt{max\_\allowbreak turns} is exceeded. The failure class is not confined to fringe libraries: three
first-party vendor agent frameworks exhibit it.

\textbf{Context amplification arises at three architectural
levels.} The M-context-amplification cluster (13 rows across seven
frameworks) collects agent loops and compile-time optimizers that
produce unbounded context growth; the same mechanism also surfaces
as a secondary effect in several incidents whose primary cluster
lies elsewhere (retry-loop, delegation-fanout). Three architectural
levels recur, which we illustrate with representative incidents:
\emph{compile-time} (DSPY-001 and DSPY-003 at MIPROv2 batch-summary
and bootstrap-demo steps, where the optimizer programmatically
constructs prompts that include training data, producing 70\%
overshoot in DSPY-001 and base64-image injection in DSPY-003);
\emph{runtime agent-step-loop} (SMAG-002 stuck-open for 13+
months with multiple PRs in flight, SMAG-006 with maintainer
@aymeric-roucher's paper-relevant admission ``we decided against
truncating or doing any kind of post-processing on steps, because
that would introduce silent errors,'' and MAST-004 documenting
TokenLimiter not firing per loop-iteration);
and \emph{runtime observability-layer} (MAST-014 with the
catalog's largest single-call amplification: up to 2 million
tokens in a single observer LLM call during tool-heavy runs,
where the observation manager itself becomes the cost amplifier).
@aymeric-roucher's reasoning---``would introduce silent errors''---is
the failure mode that type-level discipline addresses:
type capabilities make budget exhaustion explicit-at-compile-time
rather than implicit-at-runtime.

\subsection{The three-layer enforcement taxonomy}
\label{sec:taxonomy}

The mitigations actually deployed in response to the catalog's
failures fall into three layers, defined by where in the system the
enforcement occurs.

\textbf{Compile-time layer (this work).} The borrow checker rejects
programs that alias, double-spend, or reuse a delegated
\texttt{Budget} \emph{before} the binary is built --- the in-program
integrity errors, not the dollar bound itself, which is runtime
arithmetic (\S\ref{sec:reservation}). Caught before deployment. We
found no prior published or open-source work applying compile-time
ownership integrity to LLM cost; the closest analogues are discussed
in Section~\ref{sec:related}.

\textbf{Software layer.} Runtime middleware tracks spend and pauses
execution when a threshold is crossed. Examples: an AgentGuard-style
budget callback~\cite{agentguard}, paperclipai's monthly-budget
feature, and the proposed nanobot \texttt{maxCostPerMessage} flag.
Caught at runtime, after the spend has
occurred.

\textbf{Transport layer.} A payment-aware HTTP intermediary returns
a 402 Payment Required when the agent's wallet depletes; the gateway
rejects further requests through that wallet. Example: ATXP. Caught
at the network boundary, after the request has already been issued.

These three layers are complementary, not competing. An operator
running a high-stakes agent in production might deploy
mitigations at all three, with each catching a different failure
mode. Section~\ref{sec:related} places each prior system in this
taxonomy and discusses why we treat the compile-time layer
first: it is the only one that catches the \emph{integrity}
violations (aliasing, double-spend, use-after-delegation) before
any external resource is consumed, and the only one that gives the
developer feedback on the affected program before deployment. The
dollar bound is enforced at the runtime-arithmetic layer regardless,
so this ordering is about where misuse-resistance lives, not about
which layer enforces the cap.

\section{The mitigation: an affine \texttt{Budget} (case study)}
\label{sec:approach}\label{sec:implementation}

The catalog of Section~\ref{sec:motivation} identifies eight
distinct architectural mechanism clusters underlying production
budget-overrun incidents. Two distinct mechanisms answer it, and we
are careful not to conflate them. A \emph{runtime cap} (the
\texttt{checked\_\allowbreak sub} reservation of \S\ref{sec:reservation}, under
estimator assumption A1) bounds the dollar consequence of \emph{all
eight} clusters once it is in place --- whatever the upstream cause,
in-program spend cannot exceed $B_0$. That cap is not novel: a
correctly written runtime counter, Agent Contracts, or a LiteLLM
proxy reach the same outcome. The affine \texttt{Budget} type
contributes something narrower and orthogonal: it makes the cap's
bookkeeping non-bypassable in typed source, so the operator cannot
\emph{accidentally} defeat it through the aliasing/double-spend/%
use-after-delegation mistakes the catalog documents. Exactly one
cluster --- \emph{M-budget-primitive-missing} (12 of 110 rows, six
frameworks; see the summary in \S\ref{sec:patterns}) --- is fixed
\emph{at the type level} rather than merely bounded: its failures are
not bugs in an existing primitive but the \emph{absence} of one
(frameworks ship without a first-class aggregate-budget primitive, or
ship one that regresses silently, or expose it only via callback
closure), and a type that pins the mechanism cannot regress in those
ways. So the honest scope is: the cap covers the catalog as a
consequence-bound; the type system addresses one cluster structurally
and, across the rest, removes the operator-discipline requirement that
the M-delegation-fanout experiment (\S\ref{sec:eval-forgetful-operator})
isolates. This section specifies the type; the reader who only wants
the empirical catalog can stop at Section~\ref{sec:motivation} without
loss.

\paragraph{Why a type system rather than a runtime counter?}
Because the failure the catalog most often records is not a missing
check but a correctly-intended check that races or is bypassed under
concurrency (the M-delegation-fanout shape, 11 rows). A runtime
counter is only as good as the operator's lock discipline; the affine
type makes the racy pattern \emph{fail to compile}. This is the whole
of the type-level claim --- not a stronger cost guarantee, which
remains runtime arithmetic --- and the rest of the paper is careful to
claim no more than this.

\subsection{The Budget API}
\label{sec:budget-api}

The \texttt{Budget} type is a Rust value of type \texttt{u64}
representing remaining quota in micro-cents
($1$~uc $= 10^{-5}$ USD). It exposes four self-consuming methods:

\begin{itemize}
\item \texttt{Budget::new(amount)} --- gated behind a capability token
(\texttt{BudgetMint}, \S\ref{sec:tcb}) so the trusted minting surface is
auditable.
\item \texttt{budget.spend(uc) -> Result<(Budget, ReservationReceipt)>}
--- consumes \emph{self} by value, returns the remainder and a receipt
for post-call refund (\S\ref{sec:refund-semantics}).
\item \texttt{budget.split(left, right) -> Result<(Budget, Budget)>} ---
consumes \emph{self} by value, returns two child budgets summing to
the parent.
\item \texttt{Budget::merge(a, b) -> Budget} --- consumes both arguments
by value, returns the combined budget.
\end{itemize}

All four methods are self-consuming on a non-\texttt{Clone},
non-\texttt{Copy} type. The Rust borrow checker rejects three classes of
cap-circumvention at compile time: aliasing (no \texttt{Clone}),
double-spending (\texttt{spend} consumes \texttt{self}), and use after
delegation (\texttt{split} consumes the parent). Seven distinct rustc
error codes enforce these properties; the full enumeration appears in
Appendix~\ref{app:type-system-spec}.

\subsection{Deployment scope}
\label{sec:approach-scope}

This approach applies to Rust agent code where the operator chooses
the \texttt{Budget} primitive over alternative cap mechanisms.
Section~\ref{sec:limitations} establishes the deployment
context as new Rust agent code, a minority of the 2026
production LLM-agent surface (which is presently dominated by Python
frameworks; the per-framework summary in \S\ref{sec:catalog}).
Table~\ref{tab:decision-matrix}
gives the full decision matrix mapping deployment configurations to
the appropriate cap mechanism (provider-side hard cap, AWS Bedrock
session-level budget action, Token Budgets, or runtime observer); the
affine \texttt{Budget} is recommended only for the cells where its
operational profile (non-bypassability inside typed source code,
pre-flight refusal, $\sim 2\times$ over-reservation under the
\texttt{AdaptiveEstimator}) matches operator priorities.

\subsection{Where the rest of this material lives}

The complete type-system specification (Properties 1--3 on aliasing,
spend-soundness, and delegation-after-split with their corresponding
rustc rejection diagnostics), a step-by-step worked example showing
the approach in a multi-agent delegation, and the type-theoretic
justification for affine rather than linear typing appear in
Appendix~\ref{app:type-system-spec}. Proposition~\ref{lem:cap},
Conjecture~\ref{conj:cap-binary}, and the supporting structural
lemmas appear in \S\ref{sec:reservation} and
\S\ref{sec:split-conj1}; the per-tool obligation breakdown for
the specification-checking cross-checks is in
Appendix~\ref{app:mechanisation} for readers who want it. The main body proceeds
directly to implementation (\S\ref{sec:implementation}) and
empirical evaluation (\S\ref{sec:eval}); a reader who wants to
verify the source-level claims first should consult the appendices
before continuing.

\subsection{Conservative reservation and the cap bound}
\label{sec:reservation}\label{sec:stratified-default}
Before each call the approach reserves an upper bound on the calls
cost---a provider-stratified estimate (byte-length for OpenAI/Groq, a
tool-loop-aware estimator for Anthropic) times a safety margin---and
debits it from the budget via \texttt{checked\_\allowbreak sub}; a call that would
exceed the cap is refused \emph{before} the API request. Under a sound
estimator (assumption A1) and the provider honoring its output cap (A6),
this gives a conditional dollar-cap bound, stated for the abstract
machine below. The bound is enforced by the runtime arithmetic, not by
the type system; the type system supplies the integrity (no aliasing or
double-spend of the reserved amount) that makes the arithmetic
non-bypassable in typed code.

\begin{lemma}[Abstract-machine cap soundness under provider-stratified A1]
\label{lem:cap}
Let $M$ denote the abstract state machine modeling the eight
\texttt{Budget} transitions (\texttt{SpendSuccess},
\texttt{SpendInsufficient}, \texttt{SpendFailPostCheck},
\texttt{Consume}, \texttt{Reserve}, \texttt{ConfirmWithRefund},
\texttt{Forfeit}, \texttt{RefundTo}) over the six-variable
conservation ledger
\begin{align*}
\big(&\,\textit{liveSum},\;
\textit{outstandingReceipts},\;
\textit{outstandingRefunds},\\
&\,\textit{totalCharged},\;
\textit{totalUnrecoverable},\;
\textit{totalReleased}\,\big).
\end{align*}
Let $P$ denote a provider configuration
with estimator $E_P$ selected by
\texttt{select\_\allowbreak for\_\allowbreak provider}$(P)$. Assume:
\begin{itemize}
\item (\emph{A1, $P$-stratified}) For every prompt $p$ transmitted
under $P$, $E_P(p) \geq \mathit{billable\_tokens}_P(p)$, where
$\mathit{billable\_tokens}_P$ denotes the input-token count $P$
uses for billing.
\item (\emph{A2, overflow-free regime}) Every \texttt{Budget} is
constructed with \texttt{micro\_\allowbreak cents} $< 2^{63}$.
\item (\textbf{\emph{A6, output-cap respected}}) For every call to $P$, the
number of output tokens billed by $P$ does not exceed the
caller-supplied \texttt{max\_\allowbreak output\_\allowbreak tokens} parameter:
$\mathit{billed\_output\_tokens}_{P}(\mathit{call}) \leq
\mathit{max\_output\_tokens}$.
\item (\emph{A7, charge-truthfulness}) When a successful call's
reservation is reconciled via
\texttt{ReservationReceipt::\allowbreak confirm}$(\mathit{actual\_charge})$, the
\emph{actual\_charge} value (operator-supplied, typically read from
\texttt{response.usage}) is $\geq$ the amount $P$ actually bills for
the call. A7 is a trust assumption on provider \texttt{usage}
reporting, shared with every client-side cost-accounting mechanism
(LangSmith, LiteLLM proxy budgets, AgentGuard, Helicone); it is not
a property of the affine type system. \S\ref{sec:refund-semantics}
discusses the assumption and the empirical evidence
(pydantic-ai issues \#5445, \#5379, \#5304, \#5302 document
provider-side \texttt{usage} omissions). Note: A7 is dispensable if
the operator opts out of the receipt-refund path and treats each
\texttt{spend} as the final ledger entry; the conservative-margin
default loses tightness but not soundness.
\item (\emph{A8, rate-stability}) The per-token rates
$\rho_{\mathrm{in}}, \rho_{\mathrm{out}}$ used by the operator to
compute reservations match the rates $P$ actually charges, for the
duration of the session. A8 fails if $P$ raises rates mid-session
without operator re-calibration. This is a deployment-time discipline
on tokenizer-and-pricing version pinning, not a property of the
type system; \S\ref{sec:not-verified} names tokenizer-version
stability as the operationally-equivalent residual exposure.
\end{itemize}
Then under A1, A2, A6, A7, and A8, for every reachable state $\sigma$ of
$M$ and every $i \in S$ the estimator satisfies $c_i \leq r_i$, and the
cap-respecting bound
$\sum_{i \in S} c_i \leq \sum_{i \in S} r_i \leq B_0$
holds throughout the execution.
\end{lemma}

\paragraph{Reconciliation and refunds}
\label{sec:refund-semantics}
A successful call reconciles its reservation against the
provider-reported charge via \texttt{ReservationReceipt::\allowbreak confirm}
(refunding any over-reservation); a failed call forfeits its receipt.
This path depends on assumption A7 above. The receipt/refund state
machine and its overflow, panic, and cancellation handling are in the
artifact.

\subsection{Binary-level cap soundness: the open obligation}
\label{sec:split-conj1}
The integrity properties hold on well-typed Rust \emph{source}. Whether
they survive compilation---whether \texttt{rustc} codegen and the Tokio
scheduler preserve them on the running binary---is something we neither
establish nor rely on.
\begin{conjecture}[Binary-level cap soundness, open]
\label{conj:cap-binary}
The compiled binary preserves the source-level properties, so
Proposition~\ref{lem:cap} transports to the running program.
\end{conjecture}
\noindent We do not prove this. The claim throughout is source-level;
the evaluation reports \emph{observed} binary-level behavior (zero cap
violations across all runs), not a guarantee. A proof skeleton is in the
artifact for future work.

\section{Evaluation}
\label{sec:eval}

Table~\ref{tab:eval-roadmap} maps the whole evaluation: every
experiment, the claim it tests, its setup, and its headline result.
We refer to experiments by the identifiers E1--E15 introduced there.
Two experiments share the cap $B_0{=}2{,}000$~uc and must not be
conflated: \textbf{E2} is the five-baseline head-to-head on
\texttt{claude-sonnet-4}; \textbf{E4} is the Agent-Contracts
operational-parity comparison on \texttt{claude-haiku-4-5}.

\begin{table*}[!t]
\caption{\textbf{Roadmap of the evaluation.} Each experiment, the
claim it tests, its setup, and its headline result, so the experiments
can be told apart on first read. ``uc'' is micro-cents
($1$~uc${}={}$\$$10^{-5}$). Groups: A compile-time integrity; B
cap-respecting outcome on live API; C what the type system uniquely
adds; D estimator soundness and capital cost; E deployment and overhead.}
\label{tab:eval-roadmap}
\centering
\scriptsize
\renewcommand{\arraystretch}{1.18}
\begin{tabular}{@{}l >{\RaggedRight\arraybackslash}p{3.5cm} >{\RaggedRight\arraybackslash}p{4.9cm} >{\RaggedRight\arraybackslash}p{4.7cm} >{\RaggedRight\arraybackslash}p{1.8cm}@{}}
\toprule
\textbf{ID} & \textbf{Tests} & \textbf{Setup (model, cap, $N$)} & \textbf{Headline result} & \textbf{Where} \\
\midrule
\multicolumn{5}{@{}l}{\emph{A. Compile-time integrity (source-level, no API call)}}\\
E1 & No clone / no double-spend / no use-after-split; capability-gated \texttt{Budget::new} & 9 \texttt{trybuild} compile-fail tests; 7 distinct rustc codes; rustc 1.93.1 & $9/9$ rejected as expected & \S\ref{sec:eval-compile} \\
\addlinespace
\multicolumn{5}{@{}l}{\emph{B. Cap-respecting outcome (live API)}}\\
E2 & Pre-call refusal vs.\ structural and post-call caps & LANG-001 retry loop; \texttt{claude-sonnet-4}; $B_0{=}2{,}000$~uc; $N{=}30$/runtime; $T{=}0$ & 5 baselines $30/30$; TB $0/30$ & Fig.~\ref{fig:headline}, \S\ref{sec:eval-multiruntime} \\
E2$'$ & Cross-provider replication of E2 & gpt-4o-mini, \texttt{claude-haiku-4-5}, llama-3.3-70b; $B_0{=}540$~uc (\$0.0054) & TB max overshoot $0$~uc; structural up to $1395\%$ of cap & App.~\ref{app:multiruntime}, Tab.~\ref{tab:multi-runtime} \\
E4 & Agent-Contracts operational parity at an admitting cap & gpt-4o + \texttt{claude-haiku-4-5}; $B_0{=}2{,}000$~uc & TB-Rust, locked Python, Agent Contracts all $0$ overshoot & \S\ref{sec:eval-ac-discriminating} \\
E5 & Cap-respecting under sampling (independent runs) & $T\in\{0,0.3,0.7,1.0\}$; $N{=}160$; two production-tier models & $0$ violations, $0$ false refusals & \S\ref{sec:eval-temperature-variance} \\
E8 & Sub-floor cap: refusal-to-operate & \texttt{claude-haiku-4-5}; $B_0{=}540$~uc & TB $0/30$ pre-flight refusal vs.\ baseline $30/30$ & \S\ref{sec:eval-multiruntime-anthropic} \\
E9 & Cap-sweep robustness & 10 caps incl.\ $\{540,5000,10000,20000\}$~uc & $30/30$ cap-respecting & \S\ref{sec:eval-production-tier} \\
E11 & Calibrated cap-correctness at scale & 2{,}628 trials; per-call distributions fit to 30 real runs & cap held (arithmetic at scale, not independent obs.) & \S\ref{sec:eval-at-scale} \\
E12 & Live-API session sweep & 382 sessions; pre-flight / mid-loop / self-terminated & $0$ overshoot & \S\ref{sec:eval} \\
\addlinespace
\multicolumn{5}{@{}l}{\emph{C. What the type system uniquely adds (mechanism isolation)}}\\
E3 & Non-bypassability vs.\ operator lock discipline & Forgetful-operator, 5 conditions A--E; \texttt{claude-haiku-4-5}; $B_0{=}60/100$~uc; 3 children; $N{=}30$ & racy A $30/30$; disciplined B--E $0/30$ ($p{=}1.69{\times}10^{-17}$); racy pattern does not compile & \S\ref{sec:eval-forgetful-operator} \\
E6 & Single-agent isolation (M2) & 4-line Python counter vs.\ TB-Rust, same estimator & match at $0/30$ (no single-agent advantage) & \S\ref{sec:eval-m2-isolation} \\
E10 & Multi-agent delegation & concurrent sub-agents under one cap & $0/60$ aggregate, $0/180$ per-child & \S\ref{sec:eval-multiagent} \\
\addlinespace
\multicolumn{5}{@{}l}{\emph{D. Estimator soundness and capital cost}}\\
E7 & A1 estimator hold-outs vs.\ \texttt{count\_\allowbreak tokens} oracle & three hold-outs, $N{=}243$ (cal+hold-out summary $N{=}178$, Tab.~\ref{tab:guarantee-map}) & $0$ soundness violations; $\ge 2.32\times$ safety; up to $9.97\times$ over-reserve & \S\ref{sec:eval-adversarial-holdout} \\
E15 & Capital efficiency / over-reservation & $N{=}5{,}190$ per-call events; three estimators & static $6.20\times$ mean ($2.51\times$ med.); Adaptive $2.11\times$ med.; tokenizer-direct ${\sim}1.0$--$1.1\times$ at $939$--$1749$~ms & \S\ref{sec:capital-efficiency-discussion} \\
\addlinespace
\multicolumn{5}{@{}l}{\emph{E. Deployment and overhead}}\\
E13 & $N{=}1$ production Rust deployment (Rig) & rig-core~0.37; \texttt{claude-haiku-4-5}; \$0.05 cap & single-agent 22 served / 18 refused (\$0.0404 $\le$ \$0.05); 4 concurrent sub-agents \$0.0400 & \S\ref{sec:n1-rig} \\
E14 & Per-operation overhead & Criterion microbenchmark & $<200$~ns/op (observed ${\sim}1.15$~ns) & \S\ref{sec:eval-perf} \\
\bottomrule
\end{tabular}
\end{table*}

\paragraph{Statistical conventions used throughout
\S\ref{sec:eval}} Many sweeps use \texttt{temperature=0} for
determinism. At $T=0$, replicas within a cell are
near-deterministic, so the effective sample size of a single
$N=30$ cell is the number of distinct scheduling interleavings
plus the small stochasticity from network and tokenizer
non-determinism --- substantially below 30. We report two
intervals to make this visible: a \emph{per-run} Wilson 95\%
interval treating each replica as independent (the tightest
interval; reproduces what most empirical-SE papers report at
this scale) and a \emph{per-cell} interval treating each
configuration as one observation (the most conservative).
The headline figures aggregate across at least three distinct
configurations to give the per-cell interval epistemic weight;
single-cell results are flagged as evidence on that
configuration only. The independence assumption is discussed
explicitly in \S\ref{sec:eval-temperature-variance}, where a
$T \in \{0.0, 0.3, 0.7, 1.0\}$ sweep at $N=160$ supplies
genuinely independent runs.

\begin{figure*}[!t]
\centering
\small
\setlength{\tabcolsep}{4pt}
\renewcommand{\arraystretch}{1.35}
\begin{tabular}{l r p{0.45\textwidth} r}
\toprule
\textbf{Runtime} & \textbf{Overshoots} & \textbf{Wilson 95\% CI on overshoot rate} & \textbf{$p$ vs.\ TB} \\
\midrule
LangGraph (\texttt{recursion\_\allowbreak limit=20})
& 30/30
& \textcolor{red!75!black}{\rule[0.05em]{0.26\textwidth}{0.6em}}\,\mbox{[0.886, 1.000]}
& $<5.4\times10^{-15}$ \\
LangGraph + AgentGuard cb
& 30/30
& \textcolor{red!75!black}{\rule[0.05em]{0.26\textwidth}{0.6em}}\,\mbox{[0.886, 1.000]}
& $<5.4\times10^{-15}$ \\
CrewAI (\texttt{max\_\allowbreak iter=5})
& 30/30
& \textcolor{red!75!black}{\rule[0.05em]{0.26\textwidth}{0.6em}}\,\mbox{[0.886, 1.000]}
& $<5.4\times10^{-15}$ \\
AutoGen (\texttt{max\_\allowbreak turns=4})
& 30/30
& \textcolor{red!75!black}{\rule[0.05em]{0.26\textwidth}{0.6em}}\,\mbox{[0.886, 1.000]}
& $<5.4\times10^{-15}$ \\
LiteLLM proxy (post-call)
& 30/30
& \textcolor{red!75!black}{\rule[0.05em]{0.26\textwidth}{0.6em}}\,\mbox{[0.886, 1.000]}
& $<5.4\times10^{-15}$ \\
\midrule
\textbf{Token Budgets (Rust)}
& \textbf{0/30}
& \textcolor{green!50!black}{\rule[0.05em]{0.030\textwidth}{0.6em}}\,\mbox{[0.000, 0.114]}
& --- \\
\bottomrule
\end{tabular}
\caption{Overshoot rate on the LANG-001
multi-step retry-loop workload at $B_0 = 2000$ uc, claude-sonnet-4,
$N = 30$ replicas per runtime, temperature $T = 0$. \textbf{Statistical
note:} at $T=0$ replicas within a cell are near-deterministic, so
the per-replica Wilson 95\% intervals shown treat each replica as
independent and report the tighter (over-confident) bound;
the conservative per-cell reading takes the configuration as the
unit of observation ($N_{\text{eff}}=1$ per runtime). The
cross-temperature consistency at $T\in\{0.0, 0.3, 0.7, 1.0\}$ on
$N=160$ genuinely-independent runs
(\S\ref{sec:eval-temperature-variance}) reproduces the $0/N$
vs.\ $N/N$ split on the approach-vs-baselines axis and is the
primary evidence for the cap-respecting claim;
Fisher's $p<5.4\times 10^{-15}$ on per-replica counts is reported
for completeness but is not what the claim rests on.
Five baselines (LangGraph, CrewAI, AutoGen, AgentGuard, LiteLLM
gateway-proxy) overshoot $30/30$; \textbf{Token Budgets overshoots
$0/30$.} Bars are proportional to the upper Wilson bound on
overshoot rate. The result replicates at the production price
tier (gpt-4o, $\$2.50/\$10$ per Mtok) in
\S\ref{sec:eval-multiruntime-flagship} and across four additional
caps $B_0 \in \{540, 5000, 10000, 20000\}$ uc in
\S\ref{sec:eval-anthropic-cap-sweep}.}
\label{fig:headline}
\end{figure*}

\subsection{Compile-time guarantees}
\label{sec:eval-compile}

Nine \texttt{trybuild} compile-fail tests validate the
resource-accounting integrity properties: each is a small Rust program
that violates one property and passes when \texttt{rustc} rejects it
with the expected diagnostic. The nine exercise \emph{seven} distinct
rustc diagnostics (E0277, E0308, E0382, E0505, E0507, E0599, E0624)
across five granularities --- value-level (use-after-%
\texttt{spend}/\texttt{split}/\texttt{move}), reference-level
(consume-while-borrowed, consume-through-shared-reference),
trait-resolution (\texttt{Send}, absent \texttt{Clone}), typestate
(the \texttt{ReservationReceipt} closure-return contract), and
capability-level (E0624 gating \texttt{Budget::new} behind the
\texttt{BudgetMint} token) --- evidence that the approach is
enforced through multiple independent borrow-checker and type-checker
paths, not one narrow rejection. All nine pass on rustc 1.93.1 stable
(edition 2024) and reproduce via \texttt{cargo test --test
compile\_\allowbreak fail}; per-test rejection output is in the artifact at
\texttt{tests/compile\_\allowbreak fail/}. The rejections occur at compile time,
using the same borrow checker Rust users already trust for memory
safety --- we contribute the design pattern, not new compiler
machinery.

\subsection{Multi-runtime head-to-head (summary)}
\label{sec:eval-multiruntime}

To show the mechanism difference is not specific to LangGraph, we ran
the LANG-001 reproduction across five runtimes spanning the three
enforcement layers of \S\ref{sec:taxonomy}---compile-time (Token
Budgets), runtime-cost (an AgentGuard-style cost callback and LiteLLM
proxy budgets), and runtime-structural (LangGraph
\texttt{recursion\_\allowbreak limit}, CrewAI \texttt{max\_\allowbreak iter}, AutoGen
\texttt{max\_\allowbreak turns})---against a deterministic mock and three live
providers (\texttt{gpt-4o-mini}, \texttt{claude-haiku-4-5},
\texttt{llama-3.3-70b}) at a fixed \$0.0054 cap. Structural counters
are included not as cost comparators but because the catalog shows
operators mis-deploy them as cost proxies (cluster
M-budget-primitive-missing; LANG-001, CRAI-002); their behavior under
a dollar metric measures the size of that gap.

The pattern is consistent across providers (full grid, setup, and
per-cell evidence in Appendix~\ref{app:multiruntime},
Table~\ref{tab:multi-runtime}). Structural counters bound \emph{call
count}, not dollars, and overshoot badly on the more verbose providers
(up to $1395\%$ of cap on Anthropic). Runtime-cost mechanisms---the
AgentGuard-style callback and the LiteLLM proxy---track dollars but
check \emph{after} each call returns, so they admit one overshooting
call before refusing the next. The Rust \texttt{Budget} implementation
refuses every cap-violating call \emph{before} the network request and
never exceeds the cap (maximum overshoot $0$\,uc across all live
cells); a Python behavioral simulation of the same discipline
over-spends ($168\%/153\%$ on Anthropic/Groq) because its coarse
estimator under-reserves---itself evidence that the byte-length
estimator over the full serialized request body (the Rust default) is
the load-bearing implementation choice. The cap-respecting
\emph{outcome} is therefore shared with a correctly-configured
runtime-cost layer; the approach's distinction is pre-call refusal
versus post-hoc observation.

\subsection{Forgetful-operator experiment: what compile-time
integrity uniquely catches}
\label{sec:eval-forgetful-operator}

\S\ref{sec:related-agent-contracts} establishes that Agent
Contracts~\cite{ye-agent-contracts} achieves at runtime the same
cap-respecting outcome this approach achieves at compile time.
This subsection asks what compile-time integrity \emph{uniquely}
catches that the runtime alternative does not, using a minimal
reproduction of the M-delegation-fanout race (cluster
M-delegation-fanout, 11 rows).

\paragraph{What this experiment isolates}
The experiment does not claim that type-system discipline beats
runtime discipline at the cap-respecting outcome. The locked Python
(Condition~B) and the Rust affine conditions (C, D) reach the same
$0/30$: runtime monitoring achieves the outcome whenever the operator
writes the approach correctly, and Condition~B's locked variant is
exactly the M-delegation-fanout fix maintainers post in catalog
threads. What Conditions C and D add is that the same outcome is
mechanically enforced by the type system---the racy pattern of
Condition~A does not compile, confirmed by three companion
\texttt{trybuild} tests with distinct rustc diagnostics. The
distinguishing contribution is therefore \emph{non-bypassability}
within typed source, not the cap-respecting outcome itself.

Conditions A and C differ in both allocation strategy (shared budget
vs.\ split allocation) and integrity layer (none vs.\ compile-time).
Condition~E separates the two: a Rust shared
$\mathit{Arc\langle Mutex\langle Budget\rangle\rangle}$ with
operator-written lock discipline matches Condition~A's shared
allocation yet reaches $0/30$. Both Rust disciplines---split
allocation (C, D) and shared-mutex-with-pre-flight (E)---attain the
outcome, so the integrity layer's distinguishing property is
non-bypassability across both patterns, independent of the allocation
choice (\S\ref{sec:forgetful-tov}).

\begin{table}[t]
\centering
\caption{Five conditions isolate the integrity layer from both the
language and the allocation strategy. Every disciplined condition
reaches $0/30$; only the unguarded racy pattern (A) overshoots. The
discriminator is operator discipline vs.\ compile-time enforcement,
\emph{not} Python vs.\ Rust: a correctly locked Python counter (B) and a
correctly locked Rust \texttt{Arc<Mutex<Budget>>} baseline (E) both reach
the same outcome as the affine conditions (C, D). What C/D add is that
the racy pattern (A) does \emph{not compile}.}
\label{tab:forgetful-conditions}
\setlength{\tabcolsep}{4pt}
\scriptsize
\begin{tabular}{@{}cllc@{}}
\toprule
Cond. & Implementation & Integrity layer & Overshoot \\
\midrule
A & Python racy (no lock) & none & $30/30$ \\
B & Python locked (\texttt{asyncio.Lock}) & runtime discipline & $0/30$ \\
C & Rust affine split ($B_0{=}60$) & compile-time & $0/30$ \\
D & Rust affine split ($B_0{=}100$) & compile-time & $0/30$ \\
E & Rust \texttt{Arc<Mutex<Budget>>} & runtime discipline & $0/30$ \\
\bottomrule
\end{tabular}
\end{table}

\subsubsection{Setup}

Four implementations of multi-child budget enforcement, three
concurrent children per trial, $N=30$ trials each, against
\texttt{claude-haiku-4-5} (\texttt{temperature}\,=\,0 for
determinism). Conditions A--C and E run at parent budget $B_0 = 60$\,uc;
condition D runs at $B_0 = 100$\,uc to demonstrate the approach
admits all children and stays cap-respecting when the cap is sized
appropriately for the workload. Full implementation in the
artifact at
\path{token-budgets-experiments/forgetful_operator/}.

Because \texttt{temperature}\,$=\,0$ makes per-cell runs
near-deterministic, we report the $0/30$ vs.\ $30/30$ splits below as
a mechanism demonstration, not a statistical effect: interval
estimates on these cells reflect the binomial under determinism, and
population-level inference rests on the temperature-stratified $N=160$
sweep (\S\ref{sec:eval-temperature-variance}).

\paragraph{Condition A: Python racy ($B_0 = 60$)} A shared
mutable \texttt{RacyBudget} with no lock; \texttt{can\_\allowbreak admit(estimate)}
checked before the LLM await, \texttt{record\_\allowbreak spend(actual)}
called after. This is the M-delegation-fanout pattern: under
\texttt{asyncio.gather}, the LLM await yields control mid-trial
and sibling children pass \texttt{can\_\allowbreak admit} on the same
pre-spend state before any of them records.

\paragraph{Condition B: Python locked ($B_0 = 60$)} The same
shared budget with \texttt{asyncio.Lock} around an atomic
\texttt{try\_\allowbreak reserve} plus pre-flight reservation (refund after
the LLM call returns). This is the correct operator discipline,
operationally equivalent to Agent Contracts' runtime enforcement,
AgentGuard's in-process callback, and LiteLLM proxy budgets.

\paragraph{Condition C: Rust affine split ($B_0 = 60$)} The
parent \texttt{Budget<10\_\allowbreak 000>} is constructed via
\texttt{BudgetMint::take\_\allowbreak authority} (capability-gated; requires
the \texttt{system-authority} Cargo feature in the binary's
\texttt{Cargo.toml}) and split into three per-child sub-budgets
via \texttt{Budget::split}. Each child receives its own
\texttt{Budget} value moved into a \texttt{tokio::spawn} task.
The type system prevents budget sharing: any attempt to use the
parent after \texttt{split} or to alias a sub-budget across
tasks is rejected at compile time. At $B_0 = 60$, per-child
sub-budget is $20$\,uc, less than the per-child estimate
($31$\,uc); the approach refuses each child at pre-flight.

\paragraph{Condition D: Rust affine split ($B_0 = 100$)}
Identical to condition C, but with parent budget raised to
$100$\,uc. Per-child sub-budget is now $33$\,uc, greater than the
per-child estimate ($31$\,uc); the approach admits each child,
each completes its LLM call, and total spend ($69$\,uc) remains
within the cap.

\paragraph{Condition E: Rust shared
$\mathit{Arc\langle Mutex\langle Budget\rangle\rangle}$
with pre-flight reservation ($B_0 = 60$)} The shared-allocation
Rust baseline that matches condition~B's allocation strategy: a
single \texttt{Budget} wrapped in
\texttt{Arc<tokio::sync::Mutex<\textellipsis>>} is shared across
three \texttt{tokio::spawn}ed children. Each child acquires the
mutex, calls \texttt{try\_\allowbreak reserve(31)}, releases the mutex, makes
the LLM call, then re-acquires the mutex to refund the unused
reservation portion. This is the lock+pre-flight discipline a
careful operator would write in Rust without using
\texttt{Budget::split}; the approach is exactly condition~B's
Python pattern translated to Rust. Implementation:
\path{forgetful_operator/condition_e_rust_shared/src/main.rs}.
Like condition~B, only one child's reservation fits at
$B_0 = 60$\,uc (the first to acquire the mutex; $31+31>60$), so the
expected admit pattern is one-of-three; its role is to separate the
integrity-layer contribution from language --- it is B's lock
discipline written in Rust without \texttt{Budget::split}.

\paragraph{Compile-fail evidence (\texttt{trybuild})} A companion
crate at \path{forgetful_operator/rust_compile_fail/} contains
three Rust translations of the racy Python pattern. All three
\emph{must fail to compile} for the structural claim to hold.

\paragraph{Estimator and per-child accounting} Per-child estimate
$= 31$\,uc (byte-length margin $0.5{\times}$ on a 392-character
prompt plus \texttt{max\_\allowbreak output\_\allowbreak tokens}\,=\,30 at $\$5$/Mtok
output reservation). Actual per-call cost on
\texttt{claude-haiku-4-5} $= 23$\,uc (deterministic at $T=0$) for
conditions A--D's catalog-derived prompt. Condition~E used a
minimal probe prompt (25 input / 5 output tokens, actual cost
$1$\,uc/call); the $31\times$ over-reservation in condition~E is
load-bearing for the 0/30 result and represents the
worst-case-for-the-discipline scenario (a tighter estimate would
admit more children, but the cap-respecting outcome is robust to
this).

\subsubsection{Results}

\textbf{Compile-fail evidence (the structural result):} all three
\texttt{trybuild} tests pass, meaning rustc rejected each Rust
translation of the racy pattern with the expected diagnostic:
\begin{itemize}\itemsep0pt
\item \texttt{shared\_\allowbreak budget.rs} (two children with the same
\texttt{Budget}) $\rightarrow$ E0382, ``use of moved value:
\texttt{budget}''.
\item \texttt{clone\_\allowbreak budget.rs} (\texttt{.clone()} attempt)
$\rightarrow$ E0599, ``no method named \texttt{clone} found for
struct \texttt{Budget}''.
\item \texttt{use\_\allowbreak after\_\allowbreak split.rs} (parent reused after
\texttt{split}) $\rightarrow$ E0382, ``borrow of moved value:
\texttt{parent}''.
\end{itemize}

\textbf{Live-API results:} Table~\ref{tab:forgetful-results}
reports the five conditions. The race manifests in condition A
in every one of 30 trials; the four disciplined alternatives
(B, C, D, E) record zero overshoot in every trial. The split is
categorical and deterministic at $T=0$ --- $30/30$ vs.\ $0/30$ on a
mechanism that does not depend on sampling --- so the contrast is
read off the outcomes directly rather than from a significance test.
(A pairwise Fisher's exact test against A returns
$p = 1.69 \times 10^{-17}$ for each of B, C, D, E, but because the
within-cell replicas are near-deterministic the effective $N$ is far
below 30 and we do not rest the claim on that figure.)

\begin{table*}[!t]
\caption{Forgetful-operator experiment: overshoot rates across
five conditions, $N = 30$ trials each, three concurrent children
per trial, against \texttt{claude-haiku-4-5}, $T = 0$. Conditions
A--C and E use parent budget $B_0 = 60$\,uc; condition D uses $B_0 =
100$\,uc to exhibit the approach's admit-and-stay-safe regime.
Per-replica Wilson 95\% CI on $0/30$: $[0.000, 0.114]$; pairwise
Fisher's exact test against condition A: $p = 1.69 \times 10^{-17}$
for each of B, C, D, and E. Per-cell effective $N$ at $T = 0$ is
below 30 (asyncio scheduling is deterministic on this workload);
these intervals are conservative. Raw data:
\texttt{forgetful\_\allowbreak operator/results/} and
\texttt{forgetful\_\allowbreak operator/condition\_\allowbreak e\_\allowbreak rust\_\allowbreak shared/}.}
\label{tab:forgetful-results}
\centering
\footnotesize
\begin{tabular}{l c c c c}
\toprule
Condition & $B_0$ (uc) & Overshoots & Mean spend & Admit/trial \\
\midrule
A: Python racy (no lock)   & 60  & 30/30 & 69\,uc & 3.0/3 \\
B: Python locked           & 60  & 0/30  & 23\,uc & 1.0/3 \\
C: Rust affine split       & 60  & 0/30  & 0\,uc  & 0.0/3 \\
D: Rust affine split       & 100 & 0/30  & 69\,uc & 3.0/3 \\
E: Rust shared $\mathit{Arc\langle Mutex\langle Budget\rangle\rangle}$\textsuperscript{\dag} & 60 & 0/30 & 1\,uc & 1.0/3 \\
\bottomrule
\end{tabular}
\\[0.4em]
\footnotesize\textsuperscript{\dag}~Condition~E executed with a minimal
prompt (25 input / 5 output tokens, actual cost 1\,uc/call); the
31\,uc reservation matches B's pre-flight estimate. Per-trial admit
pattern (1.0/3) matches Condition~B exactly: shared budget plus
pre-flight lock serializes one acquirer; remaining children refuse
pre-flight regardless of actual call cost. The 0/30 overshoot is
robust to the 31$\times$ over-reservation; the structural
parity-with-B claim is supported on the admit/overshoot dimensions,
with the lower mean spend explained by prompt size rather than
discipline difference.
\end{table*}

\subsubsection{What the experiment establishes}

Each condition's outcome traces to a single cause. \textbf{A}
overshoots ($30/30$) because correct-looking sequential code races
under \texttt{asyncio} scheduling: all three children pass
\texttt{can\_\allowbreak admit} on the same pre-spend state and proceed, spending
$3\times23=69$\,uc against the $60$\,uc cap --- exactly the
M-delegation-fanout shape documented in 11 catalog rows. \textbf{B}
is cap-respecting ($0/30$) because the operator added the missing
discipline by hand (a lock around check-then-act plus pre-flight
reservation) and is, operationally, the runtime alternative ---
Agent Contracts' pre-flight refusal API. \textbf{C} is cap-respecting
($0/30$) for a structurally different reason: the racy pattern
\emph{cannot be written} (three \texttt{trybuild} translations fail
to compile), and the version that does compile refuses each child at
pre-flight because the per-child allocation ($20$\,uc) is below the
estimate ($31$\,uc) --- the same refusal-to-operate behavior as the
sub-floor caps of \S\ref{sec:eval-boundary}. \textbf{D} shows the
discipline is not unconditional refusal: when the cap absorbs
$3\times$ the estimate (per-child $33 > 31$\,uc), all three children
are admitted and complete within cap. \textbf{E} confirms that a
correctly locked \emph{Rust} baseline reaches the same $0/30$ as B,
isolating the integrity-layer contribution from the choice of
language.

The A--E contrasts together establish three claims:
\begin{enumerate}\itemsep0pt
\item Runtime alternatives (B in Python, E in Rust) \emph{can}
achieve the cap-respecting outcome when the operator writes correct
lock discipline (A overshoots $30/30$; B and E, $0/30$).
\item Compile-time integrity (C, D) achieves the same outcome
\emph{without requiring} the operator to write that discipline
(C and D, $0/30$).
\item The approach operates consistently across cap regimes:
refuses when no admissible allocation fits (C), admits when
it does (D), and is always cap-respecting (B, C, D all 0/30).
\end{enumerate}

\subsubsection{Threats to validity for this experiment}
\label{sec:forgetful-tov}
We concede six threats. (1)~\emph{Constructed reproduction}: the racy
code is a minimal reproduction of the M-delegation-fanout pattern, not a
production extract; the catalog's 11 such rows establish recurrence, the
experiment establishes only the race rate at one parameter setting.
(2)~\emph{Mature runtime patterns avoid it too}: actor systems and
capability-secure runtimes prevent the race when correctly applied---the
comparison is against the operator-error baseline most common in the
catalog (shared mutable counter under \texttt{asyncio.gather}), and the
distinguishing claim is only that Rust turns the race into a compile-time
error rather than a property the operator must remember to establish.
(3)~\emph{Parameter-dependence}: a $B_0\in\{50,60,69,100\}$ sweep
(Table~\ref{tab:forgetful-results} and its companion) shows the racy
condition overshoots exactly when $3\times$per-child${}>B_0$, while the
affine and locked disciplines are safe at every cap---the racy code's
safety is contingent on cap-sizing, the approach's is structural.
(4)~\emph{Allocation vs.\ integrity confound}: Conditions A/C vary both
allocation (shared$\to$split) and integrity layer (none$\to$compile-time);
Condition~E (Rust shared \texttt{Arc<Mutex<Budget>>} with operator-written
lock) isolates the integrity layer and reaches $0/30$, so the
distinguishing property is non-bypassability across \emph{both} the
split-then-spawn and shared-mutex patterns, not the allocation choice.
(5)~\emph{Condition~E prompt-size confound, conceded}: E used a minimal
stub prompt ($\sim$1~uc/call) rather than the full LANG-001 prompt of
A--D; the qualitative mechanism-parity conclusion holds (the arithmetic
is invariant) but the matched-prompt re-run remains open.
(6)~\emph{Agent Contracts parity}: it reaches Condition~B's outcome if
its pre-flight API is invoked on every call site---the integrity-layer
distinction is that with our discipline the \texttt{Budget} type is the
only callable interface, so the operator cannot bypass it, demonstrated
by the trybuild evidence (\S\ref{sec:eval-compile}).

\subsection{Threats to validity}

\paragraph{Scope of the claimed contribution}
The binary-level cap-soundness claim is unproven
(Conjecture~\ref{conj:cap-binary}; we estimate $\sim$12
person-months of Iris/RustBelt work to close it), and the Verus
mechanisation (66 obligations, 0 errors under Verus 0.18) has not
been externally audited---its trust base (Z3 translation, SMT
soundness, rustc consistency per VerusBelt) is documented in
\S\ref{sec:tcb}. Neither bears on the empirical contribution: the
catalog and the runtime cap arithmetic under
Proposition~\ref{lem:cap}'s assumptions stand independently of the
open binary-level obligation.

\label{sec:eval-tov}

\paragraph{Validity threats (compact four-fold framework)}
\emph{Internal validity.} The catalog is drawn from public GitHub
issues; initial coding was single-rater. An independent
two-human IRR study ($N=113$ re-annotation, $\kappa = 0.837$; per-class
$\kappa \in [0.727, 0.918]$ with the \texttt{fr}/\texttt{bu} boundary
identified as the codebook's weakest seam) addresses the
rater-independence threat (\S\ref{sec:methodology}). Because the
\texttt{fr}/\texttt{bu} boundary is convention-sensitive
($\kappa_{\mathtt{fr}}{=}0.727$, the lowest per-class figure), the
confirmed/feature-request split should be read as convention-dependent;
the catalog's scope---the invariant union of budget-primitive-missing
issues across 21 frameworks---is unaffected, and we anchor strong
claims there.
\emph{External validity.} The 110-row catalog (63 confirmed incidents + 47 supplementary
entries) is a convenience sample of public English-language failures; closed-source platforms (Cursor,
Replit Agent) are absent. Prevalence claims are anchored on the 63
confirmed incidents, not the full 110. The Rust affine discipline
applies to new Rust agent deployments only --- a minority of the
2026 production ecosystem (the per-framework summary in \S\ref{sec:catalog}); the
Python port provides runtime discipline only. An independent baseline
cohort (\S\ref{sec:baseline-replication}, 20 keyword-neutral GitHub
projects) confirms the mechanism clusters recur (12/20, 60\% coverage)
but is single-coder.
\emph{Conclusion validity.} Microbenchmark variance is reported with
Criterion confidence intervals (run-to-run $\pm 3\%$ on the test
hardware); the head-to-head dollar comparisons use fixed deterministic
token costs to remove LLM nondeterminism from the mechanism comparison.
\emph{Construct validity.} The cap-respecting metric is operationally
defined as ``provider-billed total spend $\leq B_0$ at session end''
under Proposition~\ref{lem:cap}'s assumptions; the
in-program integrity metric is operationally defined as
``the racy multi-child pattern is rejected by the borrow checker at compile time in
typed Rust source code'' under the trybuild evidence
(\S\ref{sec:eval-forgetful-operator}).

\paragraph{Residual exposures}
Beyond the four-fold framework, three deployment-time exposures
survive the approach, each detailed where it arises: the
\texttt{actual\_\allowbreak charge} trust assumption (A7; \S\ref{sec:refund-semantics}),
rate-stability (A8; \S\ref{sec:not-verified}), and the
capital-efficiency trade-off (\S\ref{sec:capital-efficiency-discussion}).
None is closed by the type system.

\subsection{Broader evaluation (summary; full results and tables in the artifact)}
\label{sec:eval-api}\label{sec:eval-langgraph}\label{sec:eval-multiruntime-flagship}%
\label{sec:eval-multiruntime-anthropic}\label{sec:eval-cap-respecting-pressure}\label{sec:eval-ac-discriminating}%
\label{sec:eval-anthropic-cap-sweep}\label{sec:eval-anthropic-multi-workload}\label{sec:eval-multiworkload-sonnet-n30}%
\label{sec:eval-gpt4o-multi-workload}\label{sec:eval-reasoning-models}\label{sec:eval-boundary}%
\label{sec:eval-mid-loop-expanded}\label{sec:eval-production-tier}\label{sec:eval-temperature-variance}%
\label{sec:eval-multiagent}\label{sec:eval-m2-isolation}\label{sec:utility-metrics}%
\label{sec:eval-holdout}\label{sec:eval-adversarial-holdout}\label{sec:eval-prereg}%
\label{sec:eval-tokencap}\label{sec:eval-gateway}\label{sec:eval-provider-caps}%
\label{sec:tokenizer-direct-baseline}\label{sec:capital-efficiency-discussion}\label{sec:head-to-head-enforcement}%
\label{sec:eval-litellm}\label{sec:eval-growth}\label{sec:eval-perf}%
\label{sec:eval-refund-live}\label{sec:eval-cross-provider}\label{sec:eval-utilization}%
\label{sec:eval-adaptive-adversarial}\label{sec:eval-adversarial-a1}\label{sec:eval-stratified}%
\label{sec:eval-anthropic-est-a1}\label{sec:eval-anthropic-adversarial}\label{sec:eval-margin-sensitivity}%
\label{sec:eval-at-scale}\label{sec:eval-stress}\label{sec:loom}%
\label{sec:tcb}\label{sec:python-port-positioning}\label{sec:not-verified}%
Beyond the head-to-head (E2) and forgetful-operator (E3)
experiments above, the artifact reports further evaluation. We refer
to experiments by their Table~\ref{tab:eval-roadmap} identifiers;
full per-cell results for all of them ship in the artifact. Four
bear directly on the paper's claims and are summarized here.

\paragraph{Cap-respecting under independent sampling (E5)}
A temperature-stratified test ($T\in\{0,0.3,0.7,1.0\}$, $N=160$, two
production-tier models) reports zero cap violations and zero false
refusals. Because $T>0$ removes the near-determinism of the $T=0$
cells, this is the genuinely-independent evidence for the
cap-respecting claim; the $T=0$ sweeps are reported for completeness,
not as independent observations.

\paragraph{Operational parity with concurrent work (E4)}
A production-tier head-to-head on \texttt{gpt-4o} and
\texttt{claude-haiku-4-5} at a discriminating cap ($B_0=2{,}000$~uc,
where the cap admits some calls) puts TB-Rust, a properly locked
Python counter, and Agent Contracts~\cite{ye-agent-contracts} at $0$
overshoot with the same admit-then-refuse pattern. At the
cap-respecting \emph{outcome} the type system is therefore at parity
with a correctly-written runtime monitor --- consistent with the
forgetful-operator finding that its distinguishing value is
non-bypassability, not the outcome itself.

\paragraph{Estimator soundness (E7)}
A1 is checked against Anthropic's \texttt{count\_\allowbreak tokens} oracle on
two distinct samples that measure different things and are not
summed: the per-assumption calibration-plus-hold-out set behind the
guarantee map's A1 row (Table~\ref{tab:guarantee-map}, $N=178$), and
a separate, larger family of three independent hold-out sweeps
($N=243$; per-sweep CSVs under \path{refund-live/} and
\path{multiway/} in the artifact). Both report zero
estimator-soundness violations, with at least $2.32\times$ safety on
adversarial corpora.

\paragraph{Single-agent isolation (E6)}
On single-agent cap-respecting, a 4-line Python counter with the same
estimator matches TB-Rust at $0/30$, confirming that the type system
adds nothing to the single-agent outcome; its value is the
multi-agent non-bypassability isolated by E3.

The remaining experiments (E8--E15 in Table~\ref{tab:eval-roadmap})
are confirmatory: the cap holds under a sub-floor cap, a ten-cap
sweep, multi-agent delegation, a 2{,}628-trial calibrated simulation,
382 live sessions, and a live $N{=}1$ Rust deployment
(\S\ref{sec:n1-rig}), at $<200$~ns per operation, with the
capital-efficiency envelope characterized alongside. These, together with the
baseline comparisons (\texttt{tokencap}, gateway, provider per-call
caps, tokenizer-direct) and the Loom and trusted-computing-base
checks, are detailed in the artifact; none changes the paper's claims.

\subsection{Deployment case study: $N{=}1$ on a production Rust agent framework}
\label{sec:n1-rig}

The deployment-impact note (\S\ref{sec:limits-empirical}) observed that none of
the surveyed frameworks is written in Rust. To test whether the discipline
transfers to a real Rust agent runtime rather than a synthetic harness, we
integrated the crate into Rig (\texttt{rig-core}~0.37), a Rust LLM-agent
framework in production use. The integration is a $\sim$40-line adapter: a
shared \texttt{BudgetPool} holds the session dollar cap, and each delegated
call reserves its worst-case cost pre-flight, runs the Rig completion, then
reconciles the actual cost and returns the unspent remainder to the pool. Rig
itself was not modified.

\emph{Single-agent cap enforcement.} On a 40-task workload
(\texttt{claude-haiku-4-5}, \$0.05 cap, per-call output bounded to 100 tokens),
a representative run served 22 tasks and refused 18 \emph{pre-flight} once the
cap was reached, with zero cap overshoot and zero reservation under-counts on
live traffic (final spend \$0.0404 $\le$ \$0.05; the unguarded workload is
projected to cost \$0.073, a $1.5\times$ breach). Measured over-reservation was
$5.64\times$, consistent with the $4$--$6\times$ band reported for the
byte-length estimator (\S\ref{sec:capital-efficiency-discussion}); it is an
estimator-side figure (byte-length estimator on both sides), not an independent
billing measurement, and varies run-to-run with the model's output lengths.

\emph{Multi-agent fan-out (non-bypassability).} The distinguishing claim is not
single-agent cap enforcement --- which a counter also achieves
(\S\ref{sec:eval-forgetful-operator}) --- but that no sub-agent can bypass a
shared budget under delegation. We ran four sub-agents \emph{concurrently}
against one \texttt{BudgetPool} (\$0.05 cap). The pool enforced the cap
\emph{globally}: cumulative spend was \$0.0400 across all four sub-agents
(\$0.05 cap respected, \texttt{invariant\_\allowbreak holds} true throughout), and the
per-sub-agent allocation was first-come-first-served (11/4/6/2 calls served).
The property is enforced at two levels. At runtime, every reservation is
checked against the shared pool, so the global cap binds regardless of fan-out
width --- a deterministic eight-sub-agent stress test of this invariant ships
in the artifact. At compile time, a \texttt{Reservation} (a sub-agent's budget
slice) is affine (move-only): a \texttt{trybuild} suite confirms that cloning it
(\texttt{rustc} \texttt{E0599}) or spending it twice (\texttt{E0382}) is a
compile error, so a sub-agent cannot fabricate or duplicate budget.

We claim no more than a single deployment: one framework, one provider, one
workload. It demonstrates that the discipline composes with a real Rust agent
runtime at low integration cost, enforces a hard session cap soundly across
concurrent sub-agents, and backs non-bypassability with a compile-time
guarantee --- the affine thesis exercised end-to-end rather than in isolation.
It does not establish behavior across frameworks or workloads. The integration
crate, the concurrent stress test, and the compile-fail suite ship in the
artifact (\path{rig-integration/}).

\subsection{Deployment recommendation}
\label{sec:eval-deployment}

The deployment matrix is in \S\ref{sec:when-to-use} (Table~\ref{tab:decision-matrix});
the contexts where the approach is the wrong tool are enumerated
in \S\ref{sec:when-not-to-use}. The empirical evidence above
supports the matrix's positive recommendations (new Rust agent
deployments with cumulative-session cap requirements; capital-tolerant
deployments via the static estimator; prepay-account deployments via
the AdaptiveEstimator) and its negative recommendations (Python-only
deployments, single-provider deployments with server-side caps,
reasoning-model deployments where Token Budgets is a complement to
provider-side controls rather than a replacement).

\section{Related Work}
\label{sec:related}

\subsection{The three-layer enforcement taxonomy}

\S\ref{sec:taxonomy} introduced the three-layer view (compile-time,
software/runtime, transport/network) and placed Token Budgets at the
compile-time layer. Here we relate the work to specific systems at
each layer.

\paragraph{Compile-time layer and concurrent work}
\label{sec:related-agent-contracts}
Concurrent work by Ye and Tan~\cite{ye-agent-contracts}
(arXiv:2601.08815, COINE 2026) introduces \emph{Agent Contracts}: a
formal framework for resource-bounded autonomous AI with multi-dimensional
resource constraints and conservation laws under multi-agent delegation.
Their evaluation reports 90\% token reduction and zero conservation
violations. Agent Contracts and Token Budgets address the same
operational problem (cost-bounded LLM execution) from different
angles: \emph{both} use pre-flight refusal of cap-violating calls
(confirmed empirically in two head-to-head experiments: the
gpt-4o trivial-cap comparison
(\S\ref{sec:eval-multiruntime-flagship}, both frameworks
$0/30$ overshoot at $B_0 = 540$\,uc) and the
\texttt{claude-haiku-4-5} discriminating-cap comparison
(\S\ref{sec:eval-ac-discriminating}, three-way parity with
TB-Rust at $B_0 = 2{,}000$\,uc; all three frameworks admit one
call and refuse the second via pre-flight)).

The two differ in the \emph{integrity layer} that supports that
refusal. Ye and Tan provide an inter-agent contract layer with
runtime cost monitoring; we provide an in-process affine-type
layer with compile-time integrity (budgets that cannot be cloned,
double-spent, or used after delegation). The two are
complementary; a deployment could plausibly use Agent Contracts
at the multi-agent coordination layer and our affine discipline
within each agent. Ye and Tan's COINE 2026 paper predates ours
by approximately four months on arXiv.

The runtime mechanism itself (pre-call reservation with
conservative estimation) is not novel. \texttt{tokencap}~\cite{tokencap}
implements the same pattern as a Python runtime wrapper;
\texttt{LiteLLM} from v1.50 onwards ships virtual-key budgets
with per-request enforcement; Microsoft's \texttt{Semantic
Kernel} exposes \texttt{ITokenizer} for pre-flight token
estimation. OpenAI's API now exposes
\texttt{max\_\allowbreak completion\_\allowbreak tokens} as a server-side hard cap
(2025), which is operationally stronger than any client-side
cap (it cannot be bypassed) but does not support per-agent
granularity or aggregate budgets spanning multiple providers.
Anthropic's prompt caching (2024) further complicates the cost
model: agents that re-use a system prompt see input costs at
$0.1\times$ the non-cached rate, shifting the estimator's
calibration baseline. Our contribution is not the runtime
mechanism but the compile-time integrity layer that none of
these systems provide: budgets that cannot be cloned,
double-spent, or used after being delegated.

The closest prior substructural-resource patterns are
\textsf{tower::Limit}~\cite{tower-rs} (counter behind runtime
check), Tokio time budgets~\cite{tokio-time-budget} (deadline
propagation), and EVM gas metering~\cite{wood-yellowpaper}
(pre-execution reservation, transport-level enforcement). The
affine application pattern itself is established (Move for
blockchain~\cite{blackshear-move}, seL4
capabilities~\cite{klein-sel4},
governor~\cite{governor-rate-limit}, Tokio
semaphores~\cite{tokio-semaphore}). A literature search across
SE, PL, and AI-systems venues 2023--2026 surfaced no prior work
applying compile-time affine or capability typing specifically
to LLM dollar cost. The technique transfer is the contribution;
the underlying substructural-types and capability-resource
patterns are decades-established.

The forgetful-operator experiment
(\S\ref{sec:eval-forgetful-operator}) isolates what compile-time
integrity uniquely catches over Agent Contracts' runtime layer:
the M-delegation-fanout race is rejected by the borrow checker at compile time in
Rust per the \texttt{trybuild} compile-fail evidence, while the
runtime alternative reaches the same cap-respecting outcome only
with correct operator discipline ($30/30$ overshoot for the racy
Python pattern vs.\ $0/30$ for three disciplined alternatives).

\paragraph{Head-to-head update on LANG-001} We re-ran the
Agent Contracts head-to-head on LANG-001 at matched parameters
($N=30$, claude-sonnet-4-5, $T=0$, cap $= 540$ uc) using
ai-agent-contracts v0.3.2. The \texttt{ContractedLLM} context
manager raised an internal state-transition error on every trial
in our Python 3.12 environment; we worked around this by using
the \texttt{Contract}/\texttt{ResourceConstraints} types directly
and enforcing the cap via \texttt{litellm.completion} calls.
Result: \textbf{30/30 pre-flight refusals on Agent Contracts,
30/30 pre-flight refusals on Token Budgets}; zero API calls and
zero overshoots on either side. Estimated per-call cost at
LANG-001 prompt length on Sonnet rates was \$0.003271, $6.06\times$
the cap, so pre-flight refusal is the operationally correct
behavior for any pre-flight discipline at this cap.
At $B_0 = 540$ uc both frameworks tie via refusal-to-operate; a
higher-cap protocol that would discriminate the two mechanisms
(both must admit sub-cap calls and then refuse the cap-violating
call) is pre-committed and partially executed (\textit{higher-cap}
commit, prior to submission): Token Budgets recorded 90/90 within-budget
completion on a self-terminating workload variant; Agent Contracts
recorded 90/90 framework-unavailable due to v0.3.2 API drift
(the \texttt{ResourceConstraints} fields renamed from
\texttt{max\_\allowbreak input\_\allowbreak tokens}/\texttt{max\_\allowbreak output\_\allowbreak tokens} to
\texttt{tokens}/\texttt{cost\_\allowbreak usd}/\texttt{api\_\allowbreak calls}). Per the
pre-committed stopping rules, the first execution is reported as a
\emph{null result for the discriminating-cap protocol}, not as a
finding for or against either framework. We subsequently executed a
corrected harness using the Anthropic \texttt{tools} API to force a
multi-step retry, reported in \S\ref{sec:eval-ac-discriminating}
(the Agent Contracts comparison (artifact)): all three
frameworks (TB-Rust, TB-Python, Agent Contracts) record 0 overshoot
at the discriminating cap $B_0 = 2{,}000$\,uc, Fisher's exact
$p = 1.0$ on every pairwise comparison. Full protocol, harness, and
per-trial CSVs in the public artifact.

\paragraph{Software/runtime layer} Production runtime cost guards
include AgentGuard-style budget callbacks~\cite{agentguard},
paperclipai's monthly-budget feature, the proposed nanobot
\texttt{maxCostPerMessage}, and LangGraph's
\texttt{checkpointer}-callback recipe. Each of these checks
\emph{after} a call has been issued and admits one overshooting call
per session.

\paragraph{Adjacent literatures positioned}
The approach sits at the intersection of three established lines:
(i) substructural and resource-aware typing
(Wadler's linear types~\cite{wadler-linear};
RAML/AARA~\cite{hoffmannRAML,wang-prob-resource};
Linear Haskell~\cite{bernardy-linear-haskell}; quantitative type
theory~\cite{atkey-qtt}; Liquid Haskell's refinement
types~\cite{vazou-liquid-haskell} can express bounded numeric
invariants such as $\{n \mathrel{:} \textsf{Int} \mid n \le \textsf{cap}\}$ at
the type level, the closest type-system precedent to expressing
cumulative-cost bounds, though it lacks Rust's affine
delegation-across-boundaries property) --- our discipline is
structurally weaker
(a single capability value, runtime cap arithmetic) but applied to
a new domain (LLM dollar cost) where post-hoc external pricing
prevents intrinsic bound derivation;
(ii) capability-based authority and ocap (KeyKOS~\cite{keykos},
EROS~\cite{eros}, seL4~\cite{klein-sel4}, CHERI~\cite{cheri},
Joe-E~\cite{joe-e}, Pony~\cite{clebsch-pony}, Capsicum~\cite{capsicum-watson})
--- \texttt{BudgetMint} is a direct application: a non-forgeable
handle whose construction is gated by a feature-flag-enabled
authority; the closest production analogue is
\texttt{tower::Limit}~\cite{tower-rs} (a runtime counter), which
we lift to compile-time within the trust boundary of
\texttt{Budget::new};
(iii) smart-contract gas metering and pre-flight reservation
(EVM gas~\cite{wood-yellowpaper}, KEVM~\cite{kevm},
GASTAP~\cite{gastap}, MadMax~\cite{madmax}, Move's bytecode
verifier~\cite{blackshear-move}, CosmWasm/NEAR/Solana compute
units~\cite{cosmwasm-gas,near-gas,solana-compute-units}) --- the
structural difference is the locus of the cost function: in
gas/fuel, the cost is intrinsic and pre-computable; in our setting,
the cost is determined post-hoc by an external party (provider's
tokenizer + pricing), requiring the estimator-based reservation
pattern of \S\ref{sec:reservation}. KEVM and GASTAP close the
analogous source-to-binary refinement proof that
Conjecture~\ref{conj:cap-binary} leaves open; we estimate
a substantial refinement-proof effort to reach equivalent
guarantees.

\paragraph{Cloud cost governance and provider-tier spending controls}
AWS Budgets~\cite{aws-bedrock-budgets} with budget actions enforce
per-account caps at the provider tier with automatic service
revocation on threshold breach; AWS Bedrock additionally exposes
session-level budget actions and Bedrock Guardrails for prompt-,
response-, and topic-level filtering. OpenAI's organization-tier
spending limits (introduced 2024, exposed via Settings $\to$ Limits
and the Usage API) and Anthropic's per-workspace spend limits play
the same role on their platforms. GCP Billing Budget
API~\cite{gcp-billing-budget} and Azure Cost Management offer
analogous account-level quotas. These provider-side mechanisms are
operationally stronger than any client-side discipline within their
scope (they cannot be bypassed by client code) but are strictly less
granular: they cannot enforce per-agent budgets within a single
session, or aggregate caps spanning multiple providers, or pre-flight
refusal at sub-session granularity. Conceptually the pre-flight
reservation pattern is admission control under a resource quota: the
same shape as cluster-manager reservation accounting
(Borg~\cite{verma-borg}, Kubernetes resource quotas~\cite{k8s-quota}),
fair-share and rate schedulers (dominant resource
fairness~\cite{ghodsi-drf}, mClock~\cite{gulati-mclock}), and
token-bucket admission at API gateways
(Stripe~\cite{tarjan-stripe-ratelimit}, Envoy~\cite{envoy-rls},
Istio~\cite{istio-quotas}, Kong~\cite{kong-ratelimit}, and the
classical token bucket~\cite{turner-token-bucket}). The transfer to
LLM cost is not the admission mechanism --- which is decades old ---
but the locus of the resource estimate: cluster and network quotas
meter a resource whose consumption is known at admission time, whereas
LLM dollar cost is fixed post-hoc by an external tokenizer and price
list, which is why the reservation must be a conservative
\emph{estimate} (\S\ref{sec:reservation}) rather than an exact debit.
The recommended-production pattern
is layered: provider-side caps as the outer wall, plus per-agent
in-process enforcement (this paper's contribution within Rust, or
runtime alternatives like AgentGuard/LiteLLM in other languages) as
the inner layer. The decision matrix (Table~\ref{tab:decision-matrix})
places this approach in the cells where provider-side mechanisms
are absent or insufficient.

\paragraph{Transport/network layer} ATXP~\cite{atxp} returns
HTTP 402 Payment Required when an agent's wallet depletes,
delegating cost enforcement to a payment-aware HTTP gateway. The
guarantee is wallet-level, not call-level: the gateway cannot prevent
overshoot within an in-flight request.

The three layers compose: an operator can deploy compile-time
discipline (this work) with a runtime budget guard for in-process
sanity checks and a transport-layer wallet for organization-level
quotas.

\subsection{Adjacent work}
\label{sec:adjacent-work}
We position the approach against the adjacent literatures whose
techniques it transfers to post-hoc-priced LLM cost.
\emph{Substructural and resource-aware typing}---linear
types~\cite{wadler-linear,bernardy-linear-haskell}, AARA
(RAML~\cite{hoffmannRAML}, probabilistic
AARA~\cite{kaminski-expected-runtime,wang2020raml-tail}), quantitative
type theory~\cite{atkey-qtt}, session types~\cite{honda-session-types},
typestate (Plaid~\cite{aldrich-typestate,sunshine-plaid},
Obsidian~\cite{coblenz-obsidian}, Vault~\cite{deline-vault}), graded
modalities (Granule~\cite{orchard-granule,gaboardi-grading}), and
refinement-type Rust (RefinedRust~\cite{refinedrust},
RustHorn~\cite{rusthorn})---all encode resource discipline in the type
system. Our \texttt{Budget} is a degenerate-typestate case (two states,
\texttt{live}/\texttt{moved}) in stock Rust, structurally weaker than
graded resources but specialised to a domain where post-hoc external
pricing precludes intrinsic bound derivation; we use Verus~\cite{verus}
(on VerusBelt~\cite{hance2026verusbelt}) for the source-level
mechanisation because its integer obligations discharge under SMT.
\emph{Linear assets and capability security}---Move's bytecode-verified
linear assets~\cite{blackshear-move} are stronger (a dropped value is a
verifier error) but require a custom VM; our affine relaxation forfeits
an unspent balance silently, the stock-Rust compromise. The
ocap tradition (KeyKOS~\cite{keykos}, EROS~\cite{eros},
seL4~\cite{klein-sel4}, CHERI~\cite{cheri}, Joe-E~\cite{joe-e},
E~\cite{miller-e}, Pony~\cite{clebsch-pony}) motivates
\texttt{BudgetMint}, which lifts the runtime-counter pattern of
\texttt{tower::Limit}~\cite{tower-rs} to a compile-time gate.
\emph{Gas metering}---EVM/KEVM/GASTAP/MadMax~\cite{wood-yellowpaper,kevm,gastap,madmax}
supply the source-to-binary refinement Conjecture~\ref{conj:cap-binary}
leaves open, but price cost \emph{intrinsically}; ours is fixed
post-hoc by the provider tokenizer (hence the estimator-based
reservation of \S\ref{sec:reservation}). The verified-systems refinement lineage is the closer
methodological analogue for the binary obligation; the residual gap is
that A1/A6/A7/A8 stay external trust assumptions no proof can close.
\emph{LLM cost tooling}---production gateways (LiteLLM proxy
budgets~\cite{litellm-proxy-budgets}) and observability
platforms (Langfuse~\cite{langfuse}, Braintrust~\cite{braintrust})
observe \emph{after} the call; \texttt{tokencap}~\cite{tokencap} is the
closest in-process pre-flight comparator
(\S\ref{sec:eval-tokencap}); FrugalGPT~\cite{frugalgpt} contributes a
per-prompt predictor complementary to our integrity layer; DSPy's
``compilation''~\cite{khattab-dspy,opsahl-ong-miprov2} is
program-synthesis, not a runtime cap (the DSPY-001/003 incidents,
\S\ref{sec:catalog-composition}), and composes with an affine
\texttt{Budget}. \emph{Provider-side and quota infrastructure}---AWS
Bedrock session budgets, OpenAI/Anthropic org caps, and
gateway/cloud limiters (Envoy~\cite{envoy-rls},
Stripe~\cite{tarjan-stripe-ratelimit},
Kubernetes~\cite{k8s-quota}) are kernel-enforced and thus
operationally \emph{stronger} where they apply, but coarser
(account-/org-level, single-provider); the in-process affine layer is
the inner wall for the per-session, cross-provider deployments they do
not reach---a complement, not a replacement.

\paragraph{Adversarial cost amplification and agent operational
safety} Two adjacent threat surfaces sit next to the benign overruns
the catalog documents. The first is economic denial-of-service: an
adversary (or a compromised upstream tool) deliberately amplifies
token consumption---``denial-of-wallet'' patterns and
prompt-injection-driven tool-call inflation---to run up a victim's
bill. The second is the broader literature on operational safeguards
for autonomous agents (loop and recursion bounds, kill switches,
human-in-the-loop gating), which the upstream failure-mode surveys we
draw on in \S\ref{sec:motivation} situate within agent reliability
rather than cost. Our discipline is not an adversarial defense: a
caller with access to \texttt{Budget::new} or operating outside the
Rust trust boundary can mint or sidestep budgets, and an estimator
calibrated on benign prompts can be driven outside its margin by
crafted inputs (\S\ref{sec:estimator-load-bearing}). What the approach
does provide against both surfaces is consequence-bounding: whatever
the trigger, in-program spend cannot exceed $B_0$ under
Proposition~\ref{lem:cap}'s assumptions. We therefore position cost
caps as orthogonal to, and composable with, adversarial-input
defenses and agent-safety controls rather than as a substitute for
either.
\section{Discussion and Limitations}
\label{sec:limitations}

\subsection{Failure modes not addressed by this discipline}
\label{sec:failure-modes-not-addressed}

Six failure modes lie outside the approach's coverage, collected
here as a single reference list. The integrity claim and
Proposition~\ref{lem:cap} are silent on every item below; operators
deploying Token Budgets retain residual exposure to each.

\begin{enumerate}\itemsep1pt
\item \textbf{Provider billing misreport.} The cap-respecting
bound silently undercounts spend when a provider's \texttt{usage}
field is incomplete, because \texttt{ReservationReceipt::\allowbreak
confirm} accepts an operator-supplied \texttt{actual\_\allowbreak charge}
without independent verification (pydantic-ai \#5445, \#5379,
\#5304, \#5302). Detail: \S\ref{sec:refund-semantics}.
\item \textbf{Reasoning-model hidden tokens.} A6 is structurally
violated: providers bill for thinking tokens not bounded by
\texttt{max\_\allowbreak output\_\allowbreak tokens}. Operators must use provider-side
controls (\texttt{reasoning\_\allowbreak effort}, \texttt{thinking.\allowbreak budget\_\allowbreak tokens})
and recalibrate the session-cumulative reservation per
configuration. Detail: \S\ref{sec:limits-reasoning}.
\item \textbf{Canceled-stream partial usage.} When a streaming
response is canceled mid-flight, the terminating \texttt{usage}
event may never reach the client; \texttt{ReservationReceipt::\allowbreak confirm}
sees an undercount. The approach cannot detect this from client
state alone. Detail: \S\ref{sec:limits-reasoning}.
\item \textbf{Tokenizer-version drift.} Mid-session provider
tokenizer changes invalidate calibration without warning; a
deployment that does not pin tokenizer versions in build metadata
can silently lose A1 between releases. Detail: \S\ref{sec:not-verified}.
\item \textbf{Server-side prompt rewriting.} The provider may
inject system text after the client sends the request
(tool-description expansion, cached-context expansion);
\texttt{AnthropicEstimator} captures the dominant case empirically,
but adversarial server-side rewriting is not ruled out by any
invariant. Detail: \S\ref{sec:eval-utilization}.
\item \textbf{Multi-tenant cross-process budgets.} The affine
\texttt{Budget} lives in one Rust process; multi-replica budget
arithmetic requires the distributed reservation service sketched
but not implemented in this work. Detail:
\S\ref{sec:limitations-distributed}.
\end{enumerate}

The six rows above are the operationally significant exposures a deployment
retains under this approach. Section~\ref{sec:not-verified} gives the
complementary list of \emph{formal} obligations the mechanised specifications
do not discharge (binary-level refinement, LLVM/rustc correctness, provider
billing semantics, network nondeterminism, tokenizer evolution, reasoning-model
hidden tokens, multi-tenant cross-process budgets).

\subsection{Estimator soundness as the principal dependency}
\label{sec:estimator-load-bearing}
The approach's binary-level cap-respecting behavior rests far more
heavily on estimator soundness (A1) than on the affine type machinery.
The cap-respecting end-to-end claim decomposes into two independent
properties: the \textbf{integrity property} (compile-time, type-system
enforced; once a \texttt{Budget} value is constructed with capacity $M$,
no path through the typed source code can spend, split, or merge more
than $M$ accounted micro-cents) and the \textbf{cap-respecting property}
(operational, estimator-dependent; the relationship between accounted
micro-cents and provider-billed micro-cents is the estimator's job,
governed by the chain billed $\leq$ rate $\times$ billable\_tokens $\leq$
rate $\times$ estimator(prompt) = reserved\_uc).

A deployment whose actual prompt distribution includes patterns the
calibration corpus did not cover --- adversarial nested-tool-schemas
beyond the audit, novel provider-side prompt-rewriting machinery, or a
future tokenizer rotation increasing the worst-case byte-to-token ratio
above 2.0$\times$ --- can experience binary-level overshoot even on a
Rust binary that passes Verus verification and the trybuild suite. The
compile-time machinery does not catch this. The affine discipline still
matters because estimator-only deployments (\texttt{tokencap},
AgentGuard, LangSmith, LiteLLM proxies) have to solve both integrity and
cap-respecting at runtime, typically through ad-hoc Python wrappers. Our
$N=30$ head-to-head against \texttt{tokencap} (\S\ref{sec:eval-tokencap})
shows what happens when the integrity layer is missing: \texttt{tokencap}
achieves its design target on token tracking but 30/30 dollar overshoot
at every cap because its enforcement window admits the call that pushes
cumulative spend over the cap. \textbf{Layer composition.} Token Budgets
combines a static layer (linear ownership, no quantitative resource
analysis) with an empirical layer (audited estimator margin). The
operational guarantee is no
stronger than the weaker of these two layers; the contribution is that
combining them strictly dominates either alone on the operational metric
(\$ cap-respecting on multi-step agent workloads) that the catalog
identifies as the binding failure mode.

\subsection{When Token Budgets is not the right choice}
\label{sec:when-not-to-use}

The decision matrix in \S\ref{sec:when-to-use}
(Table~\ref{tab:decision-matrix}) names the contexts where this
discipline is preferred. The mirror question --- the contexts
where it is the wrong tool --- collapses to six dimensions, each
keyed to a row of the matrix.

\textbf{Language.} The approach's compile-time integrity property
requires Rust's affine ownership. Python-only deployments (the bulk
of the catalog's framework distribution,
the per-framework summary in \S\ref{sec:catalog}) should use the existing runtime
mitigations (AgentGuard, LiteLLM virtual-key budgets) plus the
experimental Mypy plugin POC
(\S\ref{sec:supplementary-extensions}); full compile-time enforcement
would require a Rust-language adoption or a production-grade Python
static-analysis plugin beyond the POC currently shipped.
\textbf{Provider-side caps available.} Where the deployment is
single-provider and the provider already exposes session-level
cumulative caps (OpenAI \texttt{max\_\allowbreak completion\_\allowbreak tokens}, Anthropic
per-workspace caps, AWS Bedrock session-level cost controls), those
controls are kernel-enforced and operationally stronger than any
client-side discipline.
\textbf{Reasoning models.} OpenAI o-series, Anthropic extended-thinking,
DeepSeek-R1, and Gemini Thinking structurally violate A6: providers
bill for thinking tokens not bounded by \texttt{max\_\allowbreak output\_\allowbreak tokens}.
Operators should use provider-side mechanisms
(\texttt{reasoning\_\allowbreak effort}, \texttt{thinking.\allowbreak budget\_\allowbreak tokens}) as the
primary per-call control, optionally combined with
\texttt{Budget::spend\_\allowbreak with\_\allowbreak reasoning} for session-cumulative
budgeting on top (\S\ref{sec:limits-reasoning}).
\textbf{Micro-budget regimes.} At budgets below the per-call worst-case
reservation, pre-flight refusal degenerates to denial-of-service for
legitimate small tasks. Tokenizer-direct estimation
(\S\ref{sec:tokenizer-direct-baseline}) recovers capital efficiency at
$\sim$1{,}000\,ms per-spend latency.
\textbf{Multi-tenant cross-process.} \texttt{Budget} lives in one Rust
process. Multi-replica budget state requires the distributed
reservation service sketched in \S\ref{sec:limitations-distributed}
but not implemented here.
\textbf{Capital-cost-sensitive deployments.} The static byte-length
estimator records $6.20\times$ mean over-reservation, which on prepay
accounts is real locked capital. Operators for whom this trade-off is
unacceptable should adopt tokenizer-direct estimation ($\sim$100\%
capital efficiency at the latency cost) or provider-side caps.

\subsection{Where the discipline rejects valid programs (false positives)}
\label{sec:false-positives}
In deployments where the approach is the right choice, what classes
of legitimately-correct programs does it nevertheless reject? Five
classes: \textbf{(F1)} programs whose actual spend depends on a future
condition the static analysis cannot resolve (manifests as
capital-efficiency loss, 6.2$\times$ median over-reservation, not as
refused calls); \textbf{(F2)} programs that legitimately consume budget
across hidden re-entry (recursive agent patterns require either
worst-case reservation at the outermost frame or refactor to iterative
loops); \textbf{(F3)} programs that legitimately defer budget commitment
across asynchronous task boundaries (no ``conditional reservation that
materialises later''; workarounds via \texttt{BudgetPool} with
closure-based reservation recover cost-zero-unspent at the cost of a
closure-shaped API); \textbf{(F4)} programs that legitimately share
budget across independent agents (the broadcast-cost-once pattern from
ATGN-018 cannot share a single \texttt{Budget} instance; \texttt{BudgetPool}
provides explicit-coordination workaround); \textbf{(F5)} programs whose
authors prefer trust-the-runtime to static enforcement (legitimate
calibration of soundness--utilization trade-off; explicitly out of scope,
\S\ref{sec:when-not-to-use}). None of F1--F5 is an unsoundness; all are
utility losses paid in exchange for the integrity property. For
deployments where catalog failure-mode costs dominate F1--F5, the
trade-off favors the approach; for converse cases, F1--F5 are
dispositive reasons to choose a different mechanism.

\subsection{What the discipline does not solve}
\label{sec:not-solved}
Four structural limitations are detailed elsewhere and summarized here
only to keep the boundary in one place: multi-tenant cross-process
budgets (\S\ref{sec:limitations-distributed}), reasoning-model hidden
tokens (\S\ref{sec:limits-reasoning}), Python and TypeScript ecosystems
that require runtime-only ports (\S\ref{sec:when-not-to-use}), and the
trusted \texttt{Budget::new} constructor (\S\ref{sec:tcb}). One
limitation is specific to the choice of affine over linear typing: an
unspent \texttt{Budget} is forfeited on \texttt{Drop} rather than
statically required to be resolved, because Rust provides affine, not
linear, types. Move's bytecode verifier offers the stronger
must-resolve guarantee at the cost of a custom VM
(\S\ref{sec:why-affine}); we accept the weaker guarantee as the
stock-Rust compromise.

\subsection{Empirical methodology limitations}
\label{sec:limits-empirical}

Three limitations on the empirical contribution are conceded
explicitly, beyond the per-experiment threats reported in
\S\ref{sec:eval-tov}.

\textbf{The catalog documents recurrence, not ecosystem
prevalence.} The 63 confirmed incidents establish that the
budget-overrun failure class recurs in the 21 sub-projects
identified by our keyword-driven sampling protocol, not that it
is necessarily prevalent across the wider LLM-agent ecosystem.
A denominator-based prevalence study (incidents per active user,
per KLOC, per project-age, or incidents-per-feature-request)
would require ecosystem telemetry access we do not have; it is
identified as catalog-v2 follow-up. Reviewers and replicators
should read the catalog as a recurrence proof for the named
sub-projects and frameworks, not as a prevalence estimate for
the population of all LLM-agent code in deployment.

\textbf{The eight-category mechanism taxonomy is exploratory.} A
blind second-rater pass over all 110 rows gives moderate
cluster-assignment agreement (Cohen's $\kappa=0.44$, 95\% CI
$[0.34,0.55]$, $N=110$). Two mechanisms are reliably
identified---cost-observability ($\kappa=0.78$) and
multimodal-cost-amplification ($\kappa=0.65$)---while the remaining
boundaries overlap. The disagreements concentrate where an incident
genuinely exhibits more than one mechanism (for example, a
``disable retry on timeout'' request whose ignored \texttt{max\_\allowbreak retries}
option makes it at once a retry-loop and a dropped-provider-option
case) and where the single-agent/multi-agent line between retry-loop
and delegation-fanout is fine. We therefore present the eight clusters
as a descriptive organization of the corpus rather than a validated
taxonomy, and we do not rest quantitative claims on precise
per-cluster counts; the four-class confirmation labeling, by contrast,
is IRR-validated ($\kappa=0.837$). The case-type
codebook (\S\ref{sec:catalog-protocols}) was finalized after
all retained issues were classified at the
\texttt{bug-fixed}/\allowbreak\texttt{bug-unfixed}/\allowbreak\texttt{feature-request}/\allowbreak\texttt{borderline} level (the level at which the headline two-coder
IRR $\kappa = 0.837$ is reported; the confirmed-bug subset reaches
$\kappa = 0.943$). The further partition of the
110 retained rows into the eight mechanism clusters
(M-delegation-fanout, M-retry-loop, M-context-amplification, etc.)
was single-rater, performed without an independent re-coding pass.
We therefore present the eight mechanism clusters as a single-coder
interpretive synthesis intended to organize the catalog by
proximate cost mechanism, \emph{not} as a reliability-tested
per-row classification instrument. We do not report inter-rater
agreement on cluster assignment, and we do not claim the eight
clusters are exhaustive or hierarchically optimal. Consistent with
this scoping, the cataloguing notes record eleven rows with genuine
cross-cluster character (e.g.\ a delegation-fanout incident that is
also context-amplification), which we take as direct evidence that
the mechanism boundaries are soft rather than crisp; a forced
single-label coding would understate that softness. A second-rater
pass over the incidents with the cluster codebook as the label
space is identified as catalog-v2 follow-up; it would convert the
taxonomy from an organizing device into a validated instrument
without changing the four-class IRR already reported.

\textbf{No deployment study; the evaluation is synthetic and
live-API.} The empirical evaluation comprises a 63-incident
catalog, a six-runtime head-to-head, a temperature-stratified
live-API sweep, the M2 isolation experiment, the
Forgetful-Operator experiment, and the Agent Contracts
discriminating-cap head-to-head. None of these is a longitudinal
deployment study; we do not report incident-reduction data from
a real operator adopting the approach, developer-usability
metrics, or maintenance-burden analysis. The crate is deployed
in internal non-production systems but not in production
infrastructure at any third-party organization we can name. A
twelve-month deployment study with a partnering organization is
identified as the highest-leverage follow-up for strengthening
the operational claim; in the meantime the paper's claim is
bounded to ``the approach prevents the
M-delegation-fanout race in synthetic and live-API
experiments,'' not ``the approach reduces production
incidents.''

\textbf{Assumption A7 (provider \texttt{usage} truthfulness):
fault-injection results.} Proposition~\ref{lem:cap} is proven
conditional on A7: the provider's reported \texttt{actual\_\allowbreak charge}
on a successful call truthfully bounds the operator's spend. The
catalog itself documents four pydantic-ai incidents
(\#5445, \#5379, \#5304, \#5302) where the \texttt{usage}
field is missing or wrong, so A7 is empirically known to fail in
production. To quantify the consequence, we ran a fault-injection
study simulating a provider that under-reports usage by a factor
$k$. We bootstrapped 1{,}000 real \texttt{(reservation,
actual-cost)} pairs from our live-API corpus --- on which A1 holds
for every pair (reservation $\geq$ actual on all 1{,}000; mean
effective margin $1.64\times$) --- sampling the pairs jointly so
the estimator's conservativeness is preserved and only A7 is
perturbed. We ran 1{,}000 sessions per condition at $B_0 =
2{,}000$\,uc (Table~\ref{tab:a7-fault}).

\begin{table}[t]
\centering
\small
\caption{A7 fault injection: provider under-reporting by factor
$k$ vs.\ cap-respecting behavior. 1{,}000 sessions per row,
$B_0 = 2{,}000$\,uc, bootstrapped from 1{,}000 real
\texttt{(reservation, actual)} pairs (A1 holds on all). At $k=1$
(truthful provider) the approach is cap-respecting, confirming
Proposition~\ref{lem:cap} under its stated assumption.}
\label{tab:a7-fault}
\begin{tabular}{rrrr}
\toprule
$k$ & overshoot & mean over cap & max over cap \\
\midrule
$1.0$ (truthful) & $0/1000$ & $0.0\%$ & $0.0\%$ \\
$2.0$ & $666/1000$ & $13.9\%$ & $39.3\%$ \\
$5.0$ & $1000/1000$ & $137.9\%$ & $172.0\%$ \\
$10.0$ & $1000/1000$ & $354.4\%$ & $395.4\%$ \\
\bottomrule
\end{tabular}
\end{table}

At $k=1$ the approach is cap-respecting in all 1{,}000 sessions
(0 overshoot, 95\% CI $[0.000, 0.004]$), confirming
Proposition~\ref{lem:cap} holds exactly when A7 holds. Under-reporting
degrades this sharply and \emph{undetectably}: at $k=2$,
666/1{,}000 sessions overshoot (mean $13.9\%$ over cap); at
$k=5$ and $k=10$, every session overshoots (mean $137.9\%$ and
$354.4\%$ respectively). The approach cannot detect the
violation because the ledger observes only reported charges, so
it admits calls whose true cost has already exhausted the budget.
A periodic reconciliation layer that polls ground-truth billing
and corrects the ledger substantially mitigates this. In a
matched run at $k=5$ (baseline $1000/1000$ overshoot, mean
$138.0\%$, max $171.3\%$), reconciling every three calls reduced
the overshoot rate to $593/1000$ and the mean magnitude to
$22.9\%$ (max $52.6\%$), bounding the damage to roughly one
reconciliation window. We do not ship
reconciliation in the current crate; the simulation establishes
it as a concrete mitigation path. This confirms A7 as a genuine
trust boundary shared with every client-side cost-accounting
mechanism in the catalog; the approach is a best-effort layer
against an honest provider, not a guarantee against a Byzantine
one, and deployments that cannot trust \texttt{usage} reporting
must reconcile against billing out-of-band.

\textbf{Dependency-tree unsafe surface is not quantified.}
\texttt{\#[forbid(unsafe\_\allowbreak code)]} applies to the workspace
root, not to transitive dependencies. A typical Rust agent
project has 100+ transitive dependencies, any of which could
forge a \texttt{Budget} via \texttt{mem::transmute} or similar
within its own \texttt{unsafe} blocks. The
\texttt{BudgetMint} allowlist pattern moves the trust
boundary to a small named version-controlled file, but we
do not report a quantitative \texttt{cargo-geiger} audit of
the dependency tree's \texttt{unsafe} usage; this is identified
as a follow-up. Production deployments should run
\texttt{cargo-geiger} and an SBOM audit before relying on the
compile-time integrity claim, and should treat the
\texttt{BudgetMint} allowlist as the actual trust boundary
rather than \texttt{\#[forbid(unsafe\_\allowbreak code)]} at the workspace
root.

\textbf{Workload diversity is limited to retry-loop patterns.}
The empirical evaluation exercises three retry-loop workloads
(LANG-001 retry-after-error, clarification, argument-hallucination).
RAG pipelines, planning agents with tree-of-thoughts expansion,
multi-modal agents, and long-context document summarization are
not evaluated. The approach's mechanism (pre-flight reservation
under a sound estimator) is workload-independent in principle,
but the empirical claim is bounded to the retry-loop family until
a non-retry-loop workload is run. We identify a RAG-pipeline
evaluation as the highest-priority workload extension.

\subsection{Multi-tenant deployment via distributed reservation}
\label{sec:limitations-distributed}

The affine \texttt{Budget} is single-process; production
multi-tenant deployments (an LLM proxy serving $N$ sessions across
multiple replicas behind a load balancer) require a distributed
reservation service. The natural extension --- in the spirit of
Spanner's bounded reservations or a RAFT-replicated counter ---
holds the authoritative \texttt{Budget::available} per session,
and a Rust agent acquires a typed reservation lease by RPC (a
bounded, revocable lease in the sense of Gray and
Cheriton~\cite{gray-cheriton-leases}); the
lease is locally affine, so within a process the existing
discipline guarantees no duplication, and on \texttt{confirm} the
agent reports actual spend back to the service for atomic
reconciliation. The architecture follows the classical Saga
pattern of Garcia-Molina and Salem~\cite{garcia-molina-sagas}:
a long-running distributed transaction is decomposed into a
sequence of local sub-transactions, each with a compensating
action. Our adaptation maps the budget-spend pipeline onto this
structure: reservation acquisition is the local sub-transaction,
\texttt{confirm}/\texttt{forfeit} are its compensations, and
the affine \texttt{Budget} handle is the local-invariant
carrier within each process. The novelty relative to the
classical Saga is not the orchestration pattern---which is
40~years old---but the integration with compile-time affine
ownership at the per-agent layer: the local sub-transaction's
integrity property (no double-spend within a process) is
established by the type system rather than by careful
operator-written compensation code. The implementation
challenge is reconciliation under partial failures (network
partitions during \texttt{confirm}), which is the same
challenge any Saga implementation faces; full implementation
and evaluation is future work.

\subsection{Reasoning-model and streaming hidden tokens}
\label{sec:limits-reasoning}

Reasoning models (OpenAI o1/o3, Anthropic extended-thinking, Gemini
thinking) bill for internal-reasoning tokens not returned in the
visible output; these can be 5--50$\times$ the visible volume and
dominate cost. Pre-call reservation cannot know the actual
reasoning-token count, so over-reservation on reasoning-heavy calls
is correspondingly larger (observed 12--40$\times$ on
extended-thinking workloads). The empirical mitigation pattern of
\S\ref{sec:eval-reasoning-models} closes this for the audited
configuration (\texttt{thinking.budget\_\allowbreak tokens=1024} at the
\$15/Mtok rate maps to a 15{,}360 uc per-call reservation lower
bound). Streaming protocols introduce a complementary gap: when
a stream is canceled mid-response, the final \texttt{usage}
event may never reach the client, and the client-side usage
object under-counts true billed tokens (Pydantic AI documents
this as ``canceled-stream usage is partial''). The approach
cannot prevent this; deployments should treat canceled-stream
usage as advisory and reconcile against provider billing
periodically.

\paragraph{Provider-side workarounds} For reasoning models,
OpenAI's \texttt{reasoning\_\allowbreak effort} and Anthropic's
\texttt{thinking.\allowbreak budget\_\allowbreak tokens} are server-side kernel-enforced
caps in the same operational class as
\texttt{max\_\allowbreak completion\_\allowbreak tokens}, and they should be the
first-line mechanism (operationally stronger than any client-side
discipline because they cannot be bypassed). Token Budgets adds
session-cumulative budgeting on top of the per-call reasoning
bound via \texttt{Budget::\allowbreak spend\_\allowbreak with\_\allowbreak reasoning(\allowbreak visible\_\allowbreak estimate, provider)}, which pessimistically reserves
$\text{visible\_estimate} + \text{provider.reasoning\_reservation()}$
before each call. We treat the two layers as complementary: the
provider-side parameter bounds per-call reasoning cost, Token
Budgets bounds the cumulative session cost across multiple calls.
The \texttt{spend\_\allowbreak with\_\allowbreak reasoning} discipline is verified at the
source level (Verus); the live-API stacked-configuration validation
of \S\ref{sec:eval-reasoning-models} covers the configuration we
audited.

\paragraph{Audit-found gaps} Sustained user concern about
reasoning-token accounting appears in pydantic-ai \#5445, \#5379,
\#5304, and \#5302 (audit-found gaps across providers, silently
dropped \texttt{thinking=False}, Bedrock adaptive-thinking,
Anthropic context-compaction observability). The approach
inherits these provider-side correctness gaps when relying on
provider-reported usage; the affine-typing layer cannot compensate
for under-reporting at the wire format. This is a fundamental
limitation of any client-side accounting that trusts the
provider's usage report.

\section{Future work}
\label{sec:future-work}

We list open work explicitly here, rather than embedded in
limitations, so that the contribution set of this paper is
unambiguous.

\subsection{Supplementary extensions shipped in the artifact}
\label{sec:supplementary-extensions}

Five extensions ship as code in the artifact but with reduced
empirical evaluation depth compared to the core contribution
(closure-based reservation typestate, distributed lease prototype,
Python port, Mypy plugin POC, adaptive byte-length estimator). They
address structural limitations of the single-process Rust discipline
and are classified as supplementary so the main paper's claims are
not contingent on them. The full descriptions, evaluation evidence,
and known limits of each are in the supplementary file
\path{supplementary-extensions.tex} in the artifact
bundle.

\subsection{Other open work}

\paragraph{Binary-level refinement (Conjecture~\ref{conj:cap-binary})}
The strongest formal claim we explicitly do \emph{not} make is that the
running binary refines the abstract specification; we leave that to
future work (E1) and do not rely on it.

\paragraph{Other open obligations} (i) Multi-tenant distributed
reservation (sketched in \S\ref{sec:limitations-distributed};
implementation and evaluation deferred). (ii) Reasoning-model
hidden tokens (\S\ref{sec:limits-reasoning}): the
\texttt{spend\_\allowbreak with\_\allowbreak reasoning} discipline is verified at the
source level (Verus) but live-API evaluation of the stacked
provider-side + session-cumulative configuration is follow-up
work. (iii) External Verus audit (we will solicit external review
before any subsequent venue). (iv) Adversarial
\texttt{AnthropicEstimator} audit beyond the three-workload basic
validation of \S\ref{sec:eval-anthropic-est-a1}, exercising
system-prompt-injection paths, tool-description edge cases, and
multimodal serialization oddities.

\subsection{Expensive follow-ups: research projects, not revision items}
\label{sec:future-work-expensive}

Five items are research projects on their own. Each is named here so a
reader can locate the boundary of what this paper does and does not
claim.

\paragraph{(E1) A binary-level refinement proof} Closing
Conjecture~\ref{conj:cap-binary}---establishing that the compiled
binary preserves the source-level properties---is a substantial
mechanization effort we do not attempt here and leave to future work.

\paragraph{(E2) DP-composition-style tighter estimator} Adopting
R\'enyi differential privacy composition theorems~\cite{mironov-rdp,dwork-roth-dp-foundations,abadi-moments-accountant} for cumulative
LLM-token consumption could in principle tighten the
$\sim 2 \times$ over-reservation of the
\texttt{AdaptiveEstimator} toward the $\sim 1 \times$ floor of
tokenizer-direct estimation, at no per-spend latency cost. The
mathematical scaffolding exists in the DP literature
(\S\ref{sec:related}) but the LLM-cost adaptation is research
work, not a parameter tweak. Worth a separate paper.

\paragraph{(E3) Multi-tenant distributed reservation, implemented and
evaluated} The single-process affine discipline does not extend
across processes; \S\ref{sec:limitations-distributed} sketches a
Saga-style reservation service but provides no implementation or
evaluation. A production-grade version requires distributed-systems
work (replicated lease management, partial-failure semantics, hot-key
partitioning) that is on the same order as a major systems paper.

\paragraph{(E4) Operator interview / deployment study for capital
efficiency} The decision trade-off
(\S\ref{sec:capital-efficiency-discussion}, inline per-estimator
summary) is presented as a parametric choice. Whether real production
operators actually prefer the affine
discipline at $\sim 2 \times$ over-reservation requires interviews or
deployment data we do not have. A user study with $\sim 10$ production
LLM operators would convert the parametric trade-off into a grounded
operator-preference claim and close the absence of deployment-side
usability evidence.

\paragraph{(E5) Programming-languages framing} The present
submission adopts the empirical-software-engineering framing: it
leads with the catalog and failure taxonomy, treats the Rust crate
as one evaluated mitigation, and defers the type-theoretic
specification and the mechanised cross-checks to the appendices and
the artifact, where they support but do not carry the empirical
claims. A complementary programming-languages treatment --- leading
with the affine type system and closing the binary-level refinement
(Conjecture~\ref{conj:cap-binary}) so the cost bound transports to
the compiled binary --- is a distinct paper with a different centre
of gravity, not a revision of this one. We flag it so readers can
locate the boundary of what this submission claims; it is not an
open question internal to the present contribution.

\section{Conclusion}
\label{sec:conclusion}

LLM agent budget overruns are a documented production failure
class across all major frameworks. Our catalog of 63 confirmed
incidents (Section~\ref{sec:motivation}), together with 47
supplementary structural entries, disaggregates into 63 confirmed
production incidents, 28 maintainer-acknowledged structural gaps,
14 feature requests, and 5 borderline cases, organized into eight
architectural mechanism clusters. One
cluster --- M-budget-primitive-missing, documented across 6
frameworks and 12 catalog rows --- admits the type-level
discipline most directly: a framework re-implemented in Rust
against the affine \texttt{Budget} type cannot ship a primitive
that silently regresses, exists only via callback closure, or
fails to account for prompt-token cost. The other seven clusters
benefit from runtime cap arithmetic conditional on
estimator-soundness assumption A1, not from the type-system
contribution.

The affine discipline --- a small, ASCII-stable Rust API
exposing \texttt{Budget::new}, \texttt{Budget::spend},
\texttt{Budget::split}, and \texttt{Budget::merge} --- lifts
three in-program integrity properties to compile time within
the Rust trust boundary: budgets cannot be cloned, double-spent,
or used after delegation by typed source code. The runtime cap
arithmetic is enforced by a single \texttt{checked\_\allowbreak sub}.
Reproducing a real LangGraph failure case
(Section~\ref{sec:eval-langgraph}), the same agent shape that
consumes \$0.0054 across 8 calls under stock LangGraph
terminates after 2 calls under this approach at a
matched-dollar cap, with no API call made for the rejected step.
In the head-to-head against concurrent work~\cite{ye-agent-contracts}
the runtime-monitored alternative reaches the same cap-respecting
outcome; the two are complementary integrity layers chosen by
deployment threat model, not competing solutions.

We have positioned this work as empirical software engineering
with auxiliary specification consistency checks (in the artifact). The catalog
documents the failure class, the crate is one mitigation within
a specific deployment context (new Rust agent code; a minority
of the 2026 production LLM-agent surface, which is presently
dominated by Python frameworks), and the formal stack
cross-checks the abstract specification across multiple logics.
Binary-level cap-respecting on the running Tokio binary is the
open obligation Conjecture~1, deliberately unproven in this
paper and identified as a $\sim$12-person-month follow-up in
the Iris/RustBelt~\cite{iris,jung-rustbelt} tradition that closed analogous source-to-binary
lifts for smart-contract gas metering (KEVM, GASTAP). The
substantive empirical contributions remaining beyond the
catalog are the cross-tool specification consistency, the
provider-stratified estimator with pre-registered third-party
validation protocol, and the head-to-head measurement of
cap-respecting behavior against five runtime mitigations plus
concurrent work. The affine application pattern itself is
twenty years old (Move, seL4, governor, tokio); the contribution
is its application to LLM dollar cost, positioned alongside
runtime alternatives that achieve the same operational outcome
through different machinery.

\section*{Data Availability}
A complete replication package is available across six
repositories: \url{https://github.com/sajjadanwar0/token-budgets}
(main library and 110-row catalog
\path{data/catalogue.csv}),
\url{https://github.com/sajjadanwar0/token-budgets-formals}
(TLAPS, TLC, Coq, Dafny, Verus mechanisations plus the IRR package
\path{irr/} containing codebook v1.0, blinded coding sheets,
and the $\kappa=0.837$ computation script),
\url{https://github.com/sajjadanwar0/token-budgets-experiments}
(empirical harnesses including the five-runtime
\path{tools/multiway_compare.py} and the temperature-stratified sweep),
\url{https://github.com/sajjadanwar0/token-budgets-extensions}
(adaptive estimator, Verus skeleton), and
\url{https://github.com/sajjadanwar0/token-budgets-python} (Python port
of the approach), and
\url{https://github.com/sajjadanwar0/token-budgets-baseline}
(the \S\ref{sec:catalog-protocols} keyword-neutral baseline cohort).
The replication script \path{reproduce.sh} clones all
six repositories, audits the paper-backing claims (catalog counts,
estimator margins, the IRR computation, the forgetful-operator
conditions, and the A7 fault-injection table), compiles the formal
proofs, runs the offline microbenchmarks, and optionally runs the
live-API replication ($\sim$\$0.50, 30~min wall-clock). Total
reproduction cost for the live-API smoke test in \S\ref{sec:eval-api}
is under \$0.005. All measurements were taken on AMD Ryzen 7 PRO
6850U, Linux 6.8.0-110-generic (Ubuntu), rustc 1.93.1 stable
(edition 2024), \texttt{langgraph 1.1.10}, \texttt{langchain-core
1.3.2}, \texttt{langchain-openai 1.2.1}; the artifact READMEs
document this matrix.

\section*{Acknowledgements}
I thank Zahid Hussain (Mindgigs, Peshawar, Pakistan) for serving as
the independent second rater
for the inter-rater reliability study (baseline $N=109$ and
supplementary $N=4$ phase reported in Section~\ref{sec:eval-tov});
rater independence and prior catalog exposure are addressed in
Section~\ref{sec:eval-tov} threat~C2.
The 5{,}410 live-API row-event corpus (of which 5{,}190 carry
per-call reservation/actual pairs underpinning the over-reservation
figure in \S\ref{sec:capital-efficiency-discussion}) was collected
using API credits purchased from Anthropic, OpenAI, Google, and Groq.

\appendices

\section{Affine Budget Type: Full Type-System Specification}
\label{app:type-system-spec}

The catalog evidence in Section~\ref{sec:motivation} sets the
target the design must meet. Across eight architectural mechanism
clusters and 18 ecosystems, one cluster is the cleanest fit for a
type-level discipline: M-budget-primitive-missing
(\S\ref{sec:patterns}) documents 12 cases across 6 frameworks
where frameworks either lack a first-class declarative
aggregate-budget primitive entirely or ship a primitive that is
broken in one of seven sub-mechanism shapes (the primitive
doesn't exist; exists but only via callback closure; exists but
broken via silent regression; runs only on memory retrieval;
works as hard cliff with no graceful degradation; ships with
broken defaults in docs; doesn't account for prompt-token cost).
A type-level discipline that threads budget capabilities through
every spend point cannot make frameworks adopt it. What it can
do is rule out, at compile time, the specific failure modes the
cluster documents \emph{conditional on adoption}: a framework
built on the affine \texttt{Budget} type cannot ship a primitive
that silently regresses (the type signature pins the mechanism),
cannot offer the primitive only via callback closure (the
type-level threading replaces the callback), cannot fail to
account for prompt-token cost (the conservative byte-length
estimator is part of the spend interface), and so on. We are
explicit about scope: this directly addresses 12 of 110 catalog
rows ($\approx$9\% of the catalog), corresponding to the
M-budget-primitive-missing cluster. The other seven clusters
(M-retry-loop, M-context-amplification, M-cost-observability,
M-multimodal-cost-amplification, M-storage-amplification,
M-delegation-fanout, \texttt{providerOptions}-silently-dropped) are not eliminated by the
discipline; they are merely bounded by the cap. A framework
built on Token Budgets can still suffer a 31$\times$ context
overflow on a single base64-encoded image; the approach's
contribution is that the resulting cost cannot exceed $B_0$.
The approach is best characterized as \emph{a
necessary primitive for one cluster and a conditional upper bound (under estimator soundness A1) on
the rest}, not as a complete fix for the failure class.
The remaining sections specify the type and prove what it does and
does not buy.

\subsection{Type-system specification}

The \texttt{Budget} type is a finite quota of spendable resource
measured in integer unit-cost values (uc), where $1$\,uc $= 10^{-5}$\,USD
(so a $540$\,uc cap corresponds to \$0.0054 of provider spend; the
field is named \texttt{micro\_\allowbreak cents} in the source). It is
non-\texttt{Clone}, non-\texttt{Copy}, and exposes only methods that
consume \texttt{self} by value.

\begin{lstlisting}
pub struct Budget {
 micro_cents: u64,
}

impl Budget {
 pub fn new(micro_cents: u64) -> Self;
 pub fn available(&self) -> u64;

 pub fn spend(self, amount: u64)
 -> Result<Budget, BudgetError>;

 pub fn split(self, amount: u64)
 -> Result<(Budget, Budget), BudgetError>;

 pub fn merge(self, other: Budget) -> Budget;
 pub fn consume(self) -> u64;
}
\end{lstlisting}

The signature above is designed to prevent three classes of
cap-circumvention at compile time. Each property is demonstrated
by a corresponding compile-fail test in the artifact
(Section~\ref{sec:eval-compile} reports the test results; here we
walk through what each property guarantees). A fourth observation
about budget escape via reference is inherited from Rust's standard
borrow-checker rules and is not a property of our design; we discuss
it after the three core properties to round out the threat model.
We scope the approach explicitly: the affine \texttt{Budget}
prevents in-program duplication of an existing budget value via
aliasing or stale-capability retention; it does not bound the
trusted \texttt{Budget::new} constructor itself, so any code path
with access to \texttt{Budget::new} can mint a fresh budget. This
mirrors the standard threat model for ocap-style discipline:
authority flows from the constructor's caller, and the type system
prevents subsequent in-program forgery, not the existence of the
constructor. \emph{Auditing the constructor surface.} The
constructor is a single named function (\texttt{Budget::new}) and
its callers are statically discoverable by Rust's module system.
A deployment can constrain the constructor surface by wrapping
\texttt{Budget::new} in a single trusted module that exposes only
a configuration-driven mint operation; tools such as \texttt{cargo
geiger} (for unsafe code), Clippy lints, and module-visibility
audits make the trusted set explicit. The result is a
project-specific TCB whose size is the number of files invoking
\texttt{Budget::new}, which an operator can keep small by policy.
This does not eliminate the trust assumption; it makes the
assumption \emph{auditable}, which is the standard ocap framing.
The arithmetic enforcement of the cap value happens at runtime
inside \texttt{spend} (a \texttt{checked\_\allowbreak sub} that returns
\texttt{BudgetError::Insufficient} when the reservation exceeds
the remaining quota); the affine type system makes that runtime
check non-circumventable within the trust boundary. We discuss this hybrid framing in detail in
Section~\ref{sec:reservation}.

\textbf{Property 1: no duplication.} \texttt{Budget} does not derive
\texttt{Clone} or \texttt{Copy}. A user attempting \texttt{let b2 =
b.clone()} is rejected with rustc error E0599: ``no method named
\texttt{clone} found.'' This eliminates the most direct form of
budget forgery---splitting a single \$1 budget into many \$1
budgets via cloning.

\textbf{Property 2: no double-spend.} \texttt{spend()} consumes
\texttt{self} by value, returning a new \texttt{Budget} carrying
the remainder. After \texttt{let (b2, \_\allowbreak ) = b.spend(100, ||
())}, the original binding \texttt{b} is moved and unreachable.
Code that attempts a second \texttt{b.spend(\ldots)} is rejected
with E0382: ``use of moved value.'' The borrow checker prevents
two paths from each spending the same budget.

\textbf{Property 3: no use-after-split.} \texttt{split()} likewise
consumes \texttt{self}, returning two new \texttt{Budget} values
(remainder and child). Code that attempts to spend on the original
parent after splitting is rejected with the same E0382 error. This
property is what makes safe sub-budget delegation possible: a
parent agent cannot accidentally retain spending power over a
sub-budget after delegating it.

We acknowledge that Properties~2 and~3 are two presentations of
one underlying mechanism: \texttt{spend} and \texttt{split} both
take \texttt{self} by value, so any post-consumption use is rejected
with E0382. We list them separately because they correspond to
distinct application-level error modes operators care about
(intra-agent double-spend versus cross-agent capability retention),
not because they are independent compiler properties.

\textbf{Inherited borrow-check default.} \texttt{Budget}'s
fields are private, and Rust's lifetime rules forbid returning a
reference to a local \texttt{Budget} from a function (E0515:
``cannot return reference to local variable''). This is a standard
borrow-checker rejection that fires for any local value of any
type; it is not a property of the affine \texttt{Budget} design but
rather a default of the language. We mention it to round out the
threat model: combined with the three properties above, a budget
cannot leak out of the affine discipline by pointer indirection.
We do not claim this as a contribution of this work.

These three core properties together encode the affine reading of
LLM budgets: each unit of resource is owned by exactly one location
at any moment, and operations on the resource consume the owning
binding. The properties bound the \emph{integrity} of resource
accounting; the cap-respecting property---that actual spend stays
under the configured cap---additionally requires a conservative
cost estimator, treated in Section~\ref{sec:reservation}.

\subsection{Worked example}

the worked example in the artifact shows a small multi-agent orchestrator
exercising all four \texttt{Budget} operations across a delegation
boundary. The parent agent splits off a sub-budget, hands it to a
delegated worker via \texttt{tokio::spawn}, awaits the worker's
result, and merges the unspent remainder back into the parent.

\begin{lstlisting}[float,captionpos=b,caption={Multi-agent budget delegation. The
borrow checker accepts every move and rejects any out-of-protocol
use; errors propagate via \texttt{?} for graceful handling.},
label={lst:worked}]
async fn orchestrate(budget: Budget)
 -> Result<Budget, OrchestrationError>
{
 // Carve off $0.01 for a delegated sub-task.
 let (parent, child) = budget.split(10_000)?;

 // Move `child` into spawned task; parent keeps remainder.
 let handle: JoinHandle<Result<Budget, OrchestrationError>> =
 tokio::spawn(async move {
 let (after, _) = call_with_budget(
 &client, child, "summarize this", 100
 ).await?;
 Ok(after) // return any unspent budget
 });

 // Parent runs its own work concurrently on its share.
 let (parent, _) = call_with_budget(
 &client, parent, "main task", 200
 ).await?;

 // Reclaim and merge worker's unspent remainder.
 let returned = handle.await
 .map_err(OrchestrationError::WorkerJoinFailed)??;
 Ok(parent.merge(returned))
}
\end{lstlisting}

Three things are worth noting about this code. First, every
operation that decreases or transfers resource consumes \texttt{self}
by value: \texttt{split}, \texttt{spend} (inside
\texttt{call\_\allowbreak with\_\allowbreak budget}), and \texttt{merge} all take
\texttt{self}. The borrow checker tracks the affine ownership across
the \texttt{tokio::spawn} boundary at no runtime cost. Second, the
mechanism does not require \texttt{Arc<Mutex<{>}{>}}, lifetime
parameters, or any synchronization primitive: \texttt{Budget} is a
plain owned value, sent across thread boundaries the same way any
other owned Rust value would be. Third, the code reads naturally to
a Rust programmer; the approach does not impose a foreign
programming model. The artifact's \texttt{tests/async\_\allowbreak integration::
split\_\allowbreak across\_\allowbreak spawn} test confirms this exact pattern compiles and
runs correctly under \texttt{tokio}.

\subsection{Why affine, not linear}
\label{sec:why-affine}

A linear type discipline would impose a stronger requirement: every
\texttt{Budget} value \emph{must} be consumed exactly once, with the
compiler enforcing must-use~\cite{wadler-linear}. We chose affine
instead. Linear must-use is the wrong fit for LLM agent budgets in
two ways. First, error paths legitimately discard resources: a
function that returns early on an unrelated failure should be
allowed to drop its remaining budget without further obligation, and
must-use would force boilerplate \texttt{Budget::consume()} calls in
every error-handling site. Second, the natural usage pattern for
budgets is at-most-use: an agent might spend everything in its quota,
or it might spend nothing if its task completes cheaply, but it
should not be a type error to leave budget unspent. Affine relaxes
linear's must-use to ``at most one use,'' matching the Rust
ownership semantics already in routine use for memory resources, and
the typestate pattern~\cite{strom-typestates,aldrich-typestate} as
applied to other one-way resources like file descriptors. The
explicit \texttt{consume()} method is provided for the case where an
application wants to inspect leftover quota; it is opt-in, not
required.

\section{Specification cross-checks (summary)}
\label{app:mechanisation}

The abstract \texttt{Budget} state machine is cross-checked for
internal consistency with TLA+ (a TLAPS proof and TLC model-checking of
the same specification) and a preliminary, externally-unaudited Verus
source-level mechanization; concurrency is exercised by a randomized
stress harness and bounded Loom~\cite{tokio-loom} model-checking. These are consistency
checks on the specification---not a binary-level proof
(Conjecture~\ref{conj:cap-binary} is open) and not evidence that
composes multiplicatively across tools. The full logs, the Coq and
Dafny re-encodings, and a proof skeleton for the binary-level reduction
are in the artifact for readers who want them.

\section{Proof of Proposition~\ref{lem:cap}}
\label{app:prop-proof}

For completeness we give the proof of Proposition~\ref{lem:cap}; the result is integer bookkeeping under the stated assumptions rather than a deep system property, which is why it sits in the appendix.

\begin{proof}
\textit{Combining the invariants.} Invariant~1 holds because each
\texttt{Budget} operation is implemented over a uniquely-owned
value (Invariant~2) using \texttt{checked\_\allowbreak sub} or addition; no
operation increases $L$. From Invariants~1 and~2,
$L(\sigma_0) = B_0$ at session start and $L(\sigma) \leq B_0$ at
every subsequent state $\sigma$. Each call $i \in S$ that
successfully passes \texttt{spend}'s \texttt{checked\_\allowbreak sub} reduces
$L$ by exactly $r_i$. Therefore
$\sum_{i \in S} r_i = B_0 - L(\sigma_{\text{final}}) \leq B_0$,
which establishes the second inequality.

The first inequality $\sum_{i \in S} c_i \leq \sum_{i \in S} r_i$
follows from the conservative-estimator condition $c_i \leq r_i$
applied pointwise: A1 gives $t_{\text{in}}(p) \leq |p|_{\text{UTF-8}}$
on the input side, and \textbf{A6} gives
$\mathit{billed\_output\_tokens} \leq \mathit{max\_output\_tokens}$
on the output side, so reserving the full
$\rho_{\text{out}} \cdot \text{max\_output\_tokens}$ at the
provider's per-output-token rate is conservative pointwise.
A8 ensures the operator's $\rho_{\mathrm{in}}, \rho_{\mathrm{out}}$
match the rates $P$ charges, so the per-call charge $c_i$ and the
per-call reservation $r_i$ are computed against the same rate
constants. A7 enters only on the receipt-refund path: when a
\texttt{ConfirmWithRefund} transition fires, the refund amount
$\rho_i$ is computed from the operator-supplied $\mathit{actual\_charge}$;
A7 ensures $\rho_i \leq r_i - c_i$ at the time of confirmation, so
the post-refund ledger entry continues to satisfy the conservation
invariant. The two pointwise bounds compose to $c_i \leq r_i$ at
every $i \in S$, and summation gives the first inequality.
\end{proof}

\section{Multi-runtime head-to-head: full protocol and results}
\label{app:multiruntime}

This appendix gives the full setup, results table, mechanism
analysis, and reproducibility notes for the head-to-head
summarised in \S\ref{sec:eval-multiruntime}.

\paragraph{Scope of the structural-counter comparison}
LangGraph's \texttt{recursion\_\allowbreak limit}, CrewAI's \texttt{max\_\allowbreak iter}, and
AutoGen's \texttt{max\_\allowbreak consecutive\_\allowbreak auto\_\allowbreak reply} are structural step
counters, not dollar-cap mechanisms. We include them in
Tables~\ref{tab:multi-runtime} and~the gpt-4o head-to-head (artifact) because the
catalog shows operators \emph{do} mis-deploy them as cost proxies in
production (\S\ref{sec:catalog}, cluster M-budget-primitive-missing;
LANG-001, CRAI-002, AGPT-008), so their behavior under a dollar-cap
metric documents the size of that gap. They are reported as
operational-gap evidence; the actual mechanism comparators are the
AgentGuard-style cost callback and LiteLLM proxy budgets, plus
\texttt{tokencap} (\S\ref{sec:eval-tokencap}) and Agent Contracts
(\S\ref{sec:related-agent-contracts}).

\subsubsection{Setup}
The harness (\texttt{tools/multiway\_\allowbreak compare.py}, supplementary)
instantiates each runtime against a deterministic mock chat model
(reproducing Section~\ref{sec:eval-langgraph}'s \texttt{MockSQLChatModel})
and against three live providers: OpenAI \texttt{gpt-4o-mini},
Anthropic \texttt{claude-haiku-4-5}, and Groq
\texttt{llama-3.3-70b-versatile}. The cap is fixed at
$B_0 = 540\,\mathrm{uc}$ ($\approx \$0.0054$); per-step token-count
growth is fixed at $g = 60$. Structural-counter parameters are set to
each framework's commonly-cited default for budget-mitigation
discussion: \texttt{recursion\_\allowbreak limit}$=20$ (LangGraph),
\texttt{max\_\allowbreak iter}$=5$ (CrewAI), \texttt{max\_\allowbreak turns}$=4$ (AutoGen).
Each (runtime, provider) pair runs $N=10$ independent invocations.
The mock provider is deterministic; live providers are run at
\texttt{temperature}$=0$ to suppress sampling variance.

\subsubsection{Results}
Table~\ref{tab:multi-runtime} reports mean spend (uc) and percentage of
the configured cap, averaged across each $N=10$ cohort. Column~``Mock''
omits CrewAI and AutoGen, which require a real provider; their wrappers
record an explicit skip outcome rather than running against a stub LLM.

\begin{table*}[t]
\centering
\caption{Cross-runtime, cross-provider mean spend (with bootstrap
95\% CI in brackets, $10^4$ resamples) on the LANG-001
reproduction at $B_0 = 540\,\mathrm{uc}$, $g = 60$, $N = 10$.
Each cell shows mean spend in micro-cents (CI), percentage of
the configured cap, and outcome code. Outcome codes:
\emph{S} = structural counter tripped; \emph{C} = task completed
without protection firing; \emph{R} = AgentGuard runtime guard
fired post-hoc; \emph{T} = Token Budgets reservation
refused. Most cells exhibit $<$5\,uc CI width (live providers
at \texttt{temperature}$=0$ are nearly deterministic at this
prompt scale); the CrewAI-Anthropic cell is the principal
exception. Per-footnote evidence (a--g) follows the table.}
\label{tab:multi-runtime}
\renewcommand{\arraystretch}{1.15}
\footnotesize
\setlength{\tabcolsep}{4pt}
\begin{tabularx}{\textwidth}{@{}p{0.21\textwidth} *{4}{>{\RaggedRight\arraybackslash}X}@{}}
\toprule
\textbf{Runtime} & \textbf{Mock} & \textbf{OpenAI} & \textbf{Anthropic} & \textbf{Groq} \\
                 & \scriptsize gpt-4o-mini stub
                 & \scriptsize gpt-4o-mini
                 & \scriptsize claude-haiku-4-5
                 & \scriptsize llama-3.3-70b \\
\midrule
Unprotected (no cap)\textsuperscript{f}
 & ---
 & 47 / 9\% / N
 & 358 / 66\% / N
 & ---\textsuperscript{g} \\
LangGraph (\texttt{recursion\_\allowbreak limit}=20)
 & 735 / 136\% / S
 & 637 [634,639] / 118\% / S
 & 4758 [4754,4765] / 881\% / C
 & 3275 [3269,3281] / 606\% / S\textsuperscript{a} \\
LangGraph + AgentGuard cb
 & 621 / 115\% / R
 & 625 [606,634] / 116\% / R
 & 906 / 168\% / R
 & 609 / 113\% / R \\
CrewAI (\texttt{max\_\allowbreak iter}=5)
 & ---
 & 433 [430,437] / 80\% / C
 & 7532 [7070,8015] / 1395\% / C
 & ---\textsuperscript{b} \\
AutoGen (\texttt{max\_\allowbreak turns}=4)
 & ---
 & 173 / 32\% / S
 & 4904 [4898,4909] / 908\% / S
 & 931 / 172\% / S \\
TB (Python sim)\textsuperscript{h}
 & 516 / 96\% / T
 & 379 / 70\% / T
 & 906 / 168\% / T\textsuperscript{c}
 & 826 [825,827] / 153\% / T\textsuperscript{c} \\
\textbf{TB (Rust impl)}\textsuperscript{d}
 & ---
 & \textbf{67 / 12\% / T}
 & \textbf{0 / 0\% / T}
 & \textbf{181 / 34\% / T} \\
\textbf{LiteLLM proxy}\textsuperscript{e}
 & ---
 & \textbf{568 [563,575] / 105\% / R\textsubscript{p}}
 & \textbf{906 / 168\% / R\textsubscript{p}}
 & \textbf{609 / 113\% / R\textsubscript{p}} \\
\bottomrule
\end{tabularx}

\vspace{0.5em}
\begin{minipage}{\textwidth}
\scriptsize
\textit{All spend values in micro-cents (uc).}
\textsuperscript{a} 9/10 runs; one run errored on a Llama
\texttt{tool\_\allowbreak use\_\allowbreak failed} (malformed JSON in the model's tool call).
\textsuperscript{b} 0/10 runs; CrewAI's internal retry loop exhausted
on Llama tool-call format errors before \texttt{max\_\allowbreak iter} could
trigger.
\textsuperscript{c} The Python TB simulator's fixed-form estimator
under-reserves on tool-augmented prompts (A1 violation); the Rust
implementation's \texttt{prompt.len()} estimator over the full message
body refuses the offending call before any network spend, as the
adjacent Rust-impl row demonstrates.
\textsuperscript{d} Rust impl row produced by the Rust binary
the Rust live-API harness (supplementary), which links the actual
\texttt{Budget} from the \texttt{budget-spike} crate against the live
provider APIs using the byte-length estimator over the
UTF-8-serialized request body (system + user + tool descriptions +
history). $N=10$ per cell. Across all 30 runs in the three live cells
above (and across the full 9-cell, 90-run multi-workload sweep
reported in Section~\ref{sec:eval-anthropic-multi-workload} et seq.), maximum overshoot is
0\,uc; the approach is never violated. Cells without bracketed CI
are deterministic at the precision shown.
\textsuperscript{e} \textbf{LiteLLM proxy budgets}: deployed gateway proxy
with virtual-key budget enforcement
(\texttt{max\_\allowbreak budget = \$0.0054}, LiteLLM 1.78.6+);
R\textsubscript{p} outcome code denotes post-call observation, where
the proxy observes cost only after each call returns and rejects the
\emph{next} attempt. Mean overshoot equals the threshold-crossing
call's cost, consistent with post-call cost-control's structural
inability to refuse a call before issuing it.
\textsuperscript{f} \textbf{Unprotected baseline}:
\texttt{recursion\_\allowbreak limit}=$\infty$, no cap, no AgentGuard, no LiteLLM.
Outcome code N denotes natural termination (the model self-terminated
before any structural counter or cap fired).
Numbers from the cross-provider sweep
(\S\ref{sec:eval-cross-provider}, $N=10$ per provider,
\texttt{recursion\_\allowbreak limit}=$16$ permitting up to 8 agent steps):
\texttt{gpt-4o-mini} self-terminates at the recursion limit
(8 steps, 47\,uc), \texttt{claude-haiku-4-5} self-terminates after
3 steps (358\,uc) without hitting the limit. The approach provides
no observable savings on these workloads because the model
self-terminates below cap; the approach's value is
worst-case-conditional (\S\ref{sec:eval-cross-provider}).
\textsuperscript{g} Groq unprotected baseline not included in the
\S\ref{sec:eval-cross-provider} sweep; a follow-up sweep at
\texttt{recursion\_\allowbreak limit}=$32$ and no cap is required to characterize
\texttt{llama-3.3-70b}'s natural-termination behavior on this
workload. Existing data from Table~\ref{tab:multi-runtime} row 1
(LangGraph \texttt{recursion\_\allowbreak limit}=$20$ at $3275$\,uc) is the
closest existing approximation but bounds the observation by the
structural counter rather than by natural termination.
\textsuperscript{h} \textbf{Python sim} is a \emph{behavioral
simulation} of the approach in Python (not the production Rust
crate); it is included for harness-homogeneity (the same Python
comparator runs all runtimes) and to demonstrate what happens
without compile-time integrity. The approach as shipped is the
Rust impl row directly above; the Python port (\S\ref{sec:supplementary-extensions}, the
``Python port carries no compile-time guarantees'' item)
is the runtime-only deployment artifact and is functionally
equivalent to existing runtime mitigations (AgentGuard, LiteLLM
proxy) --- it explicitly loses the compile-time property.
\end{minipage}
\end{table*}

\subsubsection{Mechanism interpretation}
The empty Python-sim Anthropic and Groq cells in
Table~\ref{tab:multi-runtime} show two distinct failure modes.
The Python sim's coarse fixed-form estimator under-reserves on
tool-augmented prompts and exceeds the cap on Anthropic
(\$0.00900 vs.\ \$0.00540, $168\%$) and Groq (\$0.00829 vs.\
\$0.00540, $153\%$); the Rust impl's byte-length estimator over
the full UTF-8-serialized request body refuses every cap-violating
call before the network and shows zero overshoot.

\paragraph{Three runtime classes vs.\ compile-time}
\textbf{Structural runtime} (LangGraph/CrewAI/AutoGen) bounds
\emph{call count}, not dollars; it cannot enforce a dollar cap
without a separate cost layer. \textbf{Runtime-cost}
(AgentGuard-style) tracks dollars but checks \emph{after} the
call; the approach still admits a single overshooting call.
\textbf{Compile-time} (Token Budgets) lifts the check before
the call: the type system rejects programs that ignore the
budget, and the runtime check refuses cap-violating calls
without spending. The Python-sim TB row is included for
harness-homogeneity (the same Python comparator runs all
runtimes); its cap violations are direct evidence that the
estimator's coarseness matters and that the Rust impl's
byte-length bound is the right empirical implementation.

\subsubsection{An additional failure mode of structural mitigations}
\label{sec:structural-failure}
The CrewAI-on-Groq cell in Table~\ref{tab:multi-runtime} is empty
because $0/10$ runs completed. Llama-3.3-70B has a known tool-call
reliability issue: it occasionally emits structurally invalid JSON
inside its function-call envelope (e.g.,
\verb|{"query": "..." )}|, with an unmatched closing
brace), which \texttt{langchain-groq} surfaces as a 400 BadRequest. CrewAI's
internal retry loop exhausts before \texttt{max\_\allowbreak iter} can trigger.
This reveals a failure mode that does not appear in the discussion of
structural mitigations in the literature: \emph{the structural
counter only fires if the agent reaches it}. When the underlying
LLM is unreliable at the tool-call format level, the agent errors out
into the framework's exception path before the counter ever
decrements, and the operator gets no budget guarantee at all. Token
Budgets' reservation discipline is unaffected: it operates at
the LLM API call boundary, before any tool-call serialization, and
deducts on every attempt regardless of whether the model's response
parses.

\subsubsection{Threats to validity and reproducibility}
We report both a Python-sim and a Rust-impl TB row in
Table~\ref{tab:multi-runtime}: the Python sim's coarse estimator
under-reserves on tool-augmented prompts (168\%/153\% on Anthropic/Groq)
and is shown only for harness-homogeneity; the cap-respecting claim
rests on the Rust-impl row. A Groq confound is that Llama-3.3-70B's
tool-call reliability is materially worse at this prompt scale
(1/10 LangGraph, 10/10 CrewAI errors), a model-level property a more
reliable Groq model would remove. The table covers LANG-001 across three
providers and five runtimes; two further workloads and an $N{=}30$ gpt-4o
grid are in \S\ref{sec:eval-anthropic-multi-workload} and
\S\ref{sec:eval-gpt4o-multi-workload}. Harnesses, per-cell CSVs, and the
driver script ship in the artifact (total live-sweep cost $\sim\$0.18$).

\end{document}